\documentclass[twocolumn,10pt]{IEEEtran}
\usepackage{bm,cite,float,amsmath,amssymb,amsthm}

\usepackage{algorithm}
\usepackage[noend]{algpseudocode}
\usepackage{indentfirst}
\makeatletter
\newcommand{\vast}{\bBigg@{4}}
\newcommand{\Vast}{\bBigg@{5}}
\makeatother

\usepackage{lipsum}
\usepackage{mathtools}
\usepackage{cuted}
\usepackage[table]{xcolor}
\makeatletter
\def\BState{\State\hskip-\ALG@thistlm}
\makeatother

\usepackage{amssymb}
\usepackage{amsmath}
\usepackage{graphicx}
\usepackage{cite}

\usepackage{balance}
\usepackage[utf8]{inputenc}

\newtheorem{lemma}{Lemma}

\IEEEoverridecommandlockouts

\usepackage{graphicx,epstopdf}
\usepackage{epsfig}	
\usepackage{amsfonts,balance}
\usepackage{bbm}
\floatname{algorithm}{Algorithm}

\newcommand{\subparagraph}{}

\usepackage{lipsum}

\newtheorem{theorem}{Theorem}

\newtheorem{corollary}{Corollary}

\makeatletter
\def\ScaleIfNeeded{%
	\ifdim\Gin@nat@width>\linewidth \linewidth \else \Gin@nat@width
	\fi } \makeatother

\makeatletter
\def\underbracex#1#2{\mathop{\vtop{\m@th\ialign{##\crcr
				$\hfil\displaystyle{#2}\hfil$\crcr
				\noalign{\kern3\p@\nointerlineskip}%
				#1\crcr\noalign{\kern3\p@}}}}\limits}

\def\underbracea{\underbracex\upbracefilla}

\def\upbracefilla{$\m@th \setbox\z@\hbox{$\braceld$}%
	\bracelu\leaders\vrule \@height\ht\z@ \@depth\z@\hfill 
	\kern\p@\vrule \@width\p@\kern\p@\vrule \@width\p@\kern\p@\vrule \@width\p@
	$}

\def\upbracefillb{$\m@th \setbox\z@\hbox{$\braceld$}%
	\vrule \@width\p@\kern\p@\vrule \@width\p@\kern\p@\vrule \@width\p@\kern\p@
	\leaders\vrule \@height\ht\z@ \@depth\z@\hfill\bracerd
	\braceld\leaders\vrule \@height\ht\z@ \@depth\z@\hfill
	\kern\p@\vrule \@width\p@\kern\p@\vrule \@width\p@\kern\p@\vrule \@width\p@
	$}

\def\underbracec{\underbracex\upbracefillc}

\def\upbracefillc{$\m@th \setbox\z@\hbox{$\braceld$}%
	\vrule \@width\p@\kern\p@\vrule \@width\p@\kern\p@\vrule \@width\p@\kern\p@
	\leaders\vrule \@height\ht\z@ \@depth\z@\hfill
	\kern\p@\vrule \@width\p@\kern\p@\vrule \@width\p@\kern\p@\vrule \@width\p@
	$}

\def\underbraced{\underbracex\upbracefilld}
\def\upbracefilld{$\m@th \setbox\z@\hbox{$\braceld$}%
	\vrule \@width\p@\kern\p@\vrule \@width\p@\kern\p@\vrule \@width\p@\kern\p@
	\leaders\vrule \@height\ht\z@ \@depth\z@\hfill\braceru$}

\def\upbracefillbd{$\m@th \setbox\z@\hbox{$\braceld$}%
	\vrule \@width\p@\kern\p@\vrule \@width\p@\kern\p@\vrule \@width\p@\kern\p@
	\bracerd\braceld
	\leaders\vrule \@height\ht\z@ \@depth\z@\hfill\braceru$}

\makeatother

\begin{document}
	
	%
	\title{Analyzing Novel Grant-Based and Grant-Free Access Schemes for Small Data Transmission}
	\author{Hui~Zhou,~\IEEEmembership{Student Member,~IEEE,}
		Yansha~Deng,~\IEEEmembership{Member,~IEEE,}
		Luca~Feltrin,
		Andreas~Höglund.
		
		\vspace{-0.2cm}
		\thanks{This work was supported by Engineering and Physical Sciences Research Council (EPSRC), U.K., under Grant EP/W004348/1. (Corresponding author: Yansha Deng)}
		\thanks{
			H. Zhou and Y. Deng are with  Department of Engineering, King's College London, London, WC2R 2LS, UK (email:\{hui.zhou, yansha.deng\}@kcl.ac.uk).}
		\thanks{L. Feltrin and A. Höglund are with the Ericsson AB (email:\{luca.feltrin, andreas.hoglund\}@ericsson.com).}
		
	}
	\maketitle

	\begin{abstract}

		Fifth Generation (5G) New Radio (NR) does not support data transmission during random access (RA) procedures, which results in unnecessary control signalling overhead and power consumption, especially for small data transmission (SDT). Motivated by this, 3GPP has proposed 4/2-step SDT RA schemes based on the existing grant-based (4-step) and grant-free (2-step) RA schemes, with the aim to enable data transmission during RA procedures in Radio Resource Control (RRC) Inactive state. To compare the 4/2-step SDT RA schemes with the benchmark 4/2-step RA schemes, we provide a spatio-temporal analytical framework to evaluate the RA schemes, which jointly models the preamble detection, Physical Uplink Shared Channel (PUSCH) decoding, and data transmission procedures. Based on this analytical model, we derive the analytical expressions for the overall packet transmission success probability and average throughput in each RACH attempt. We also derive the average energy consumption in each RACH attempt. Our results show that 2-step SDT RA scheme provides the highest overall packet transmission success probability, and the lowest average energy consumption, but the performance gain decreases with the increase of device intensity.
		
	\end{abstract}
	
	
	\begin{IEEEkeywords}
		Grant-based, Grant-free, 4-step, 2-step, Small data, and Energy consumption.
	\end{IEEEkeywords}

	%
	\maketitle

	\section{Introduction}
	As an emerging technology, the Internet of Things (IoT) enables physical objects (e.g., sensors) to be connected to the Internet, which has been implemented through massive Machine Type Communications (mMTC), and Ultra Reliable Low Latency Communications (URLLC) services in the Fifth Generation (5G) New Radio (NR). The mMTC provides ubiquitous connections, and URLLC guarantees stringent constraints of reliability and latency\cite{NR-requirements}. A plethora of applications, including unmanned aerial vehicle (UAV), wearable devices, industrial wireless sensor networks (IWSN), smart meters, and etc, are being revolutionized via IoT, in which small data packets are the typical form of traffic generated by IoT devices (e.g., hundreds of bits)\cite{Wang2019}. In view of this, minimizing control signalling overhead becomes a critical issue during small data transmission (SDT) due to the following two reasons: 1) the control signalling is non-negligible compared to the small data packets\cite{Durisi2016}; 2) non-negligible control signalling results in unnecessary energy consumption and degrades the battery performance of IoT devices, which is considered to be the main problem for IoT devices with limited battery capacity\cite{Han2019}. 
	
	\begin{figure}[htbp!]
		\centering
		\includegraphics[scale=0.8]{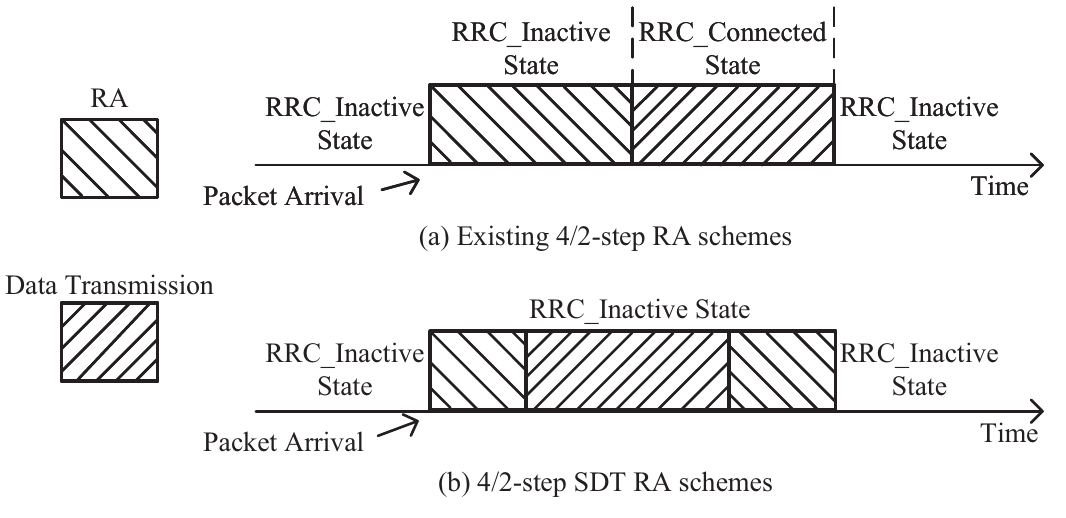}
		\caption{Uplink transmission of (a) existing 4/2-step RA schemes; and (b) newly proposed 4/2-step SDT RA schemes.}
		\label{fig:RRC State Transition}
	\end{figure}
	
	The IoT device is triggered to perform random access (RA) when a new packet arrives, and transits to Radio Resource Control (RRC) Connected state for data transmission in existing 4-step and 2-step RA schemes as shown in Fig.~\ref{fig:RRC State Transition}(a). The 4-step RA scheme utilizes a scheduling-based transmission mechanism, namely grant-based (GB) access, in which the Physical Uplink Shared Channel (PUSCH) resource is allocated to the device only after the base station (BS) receives the preamble\cite{Jang2017}. The 2-step RA scheme, introduced in 3GPP Release 16, was initially developed for the unlicensed spectrum to reduce the number of Listen Before Talk (LBT) to perform, and it also reduces the latency for RA procedures\cite{3gpp-2step}. The 2-step RA scheme is categorized as grant-free (GF) access, where the PUSCH resource is pre-allocated and transmitted along with the preamble. 
	
	RRC state transitions in the existing 4/2-step RA schemes result in unnecessary control signalling overhead, especially for SDT \cite{Hailu2019,3GPP2018}. To optimize the support for SDT, 3GPP has identified the data transmission within 4/2-step RA schemes as the potential solutions for NR in Release 17 Work Item \cite{SDT_WID}, namely 4-step SDT and 2-step SDT RA schemes, where the idea of transmitting data during RA procedure was first specified for Narrowband Internet of Things (NB-IoT) and Long Term Evolution for Machine Type Communication (LTE-M) in 3GPP Release 15. Compared to the existing 4/2 step RA schemes, the 4/2-step SDT RA schemes, as shown in Fig.~\ref{fig:RRC State Transition}(b), have the potential capability to reduce the controlling overheads and energy consumption via enabling data transmission during RRC Inactive State.
	
Recent works in \cite{Thota2019, nokia-eval, Ericsson_eval, Hoglund2018, Khlass2021} have evaluated the 4-step, 2-step, 4-step SDT, and 2-step SDT RA schemes via system-level simulation. In \cite{Thota2019}, the authors proposed three RA enhancements based on the 2-step RA scheme to reduce the collision probability, which include Parallel Preamble Transmissions, Enhanced Back-off, and Dynamic Reserved Preambles. In \cite{nokia-eval}, the preamble detection and PUSCH resource decoding was analyzed via link level simulation for different control plane payload sizes, where different Demodulation Reference Signal (DMRS) sequences were applied to handle potential PUSCH resource unit (PRU) collision. In \cite{Ericsson_eval}, considering that the GF scheme is limited to small cell size due to unsynchronized uplink, the coverage extension of 2-step RA scheme with different PUSCH size transmissions was exploited via system level simulation. In \cite{Hoglund2018}, the authors evaluated the battery life and latency of the 4-step SDT scheme, where three traffic models with various payload sizes and periods were considered. In \cite{Khlass2021}, the authors compared the performance of 4/2-step SDT RA schemes with the traditional 4/2-step RA schemes under different mobility patter, traffic pattern, and packet size. However, the general RA model and mathematical framework for the 4-step, 2-step, 4-step SDT, and 2-step SDT RA schemes have never been fully established, and their comparative insights have not been investigated yet.

	To characterize and analyze the performance of RA schemes, stochastic geometry has been regarded as a powerful tool to capture the uncertainty of devices' locations in wireless networks\cite{Haenggi2012}, which has been utilized to analyse the 4/2-step RA schemes \cite{Gharbieh2017,Jiang2018,Jiang2018a,Jiang2018b, Gharbieh2018,Liu2020}. In \cite{Gharbieh2017}, the authors analyzed the Signal to Noise plus Interference Ratio (SINR) outage of the 4-step RA scheme with three different preamble transmission schemes, where the mutual interference was considered. In \cite{Jiang2018}, the authors analyzed the queue evolution of the 4-step RA scheme by developing a spatio-temporal mathematical framework, in which preamble transmission success probability was derived. The work in \cite{Jiang2018a} extended the preamble transmission probability analysis to the preamble collision of the 4-step RA scheme, where only the device succeeds in collision will receive the granted resource for data transmission. The work in \cite{Jiang2018b} modelled the 4-step RA scheme in the NB-IoT network under time-correlated interference by considering the repeated preamble transmission and collision. The work in \cite{Gharbieh2018} analyzed both 4/2-step RA schemes based on stochastic geometry and queueing theory, where the device with the largest SINR succeeds in the collision. The work in \cite{Liu2020} defined the latent access failure probability of the 2-step RA scheme in URLLC service, and proposed a tractable approach to analyze the 2-step RA scheme under three different hybrid automatic repeat request (HARQ) mechanisms, including Reactive, K-repetition, and Proactive schemes. However, to the best of our knowledge, existing works have either focused on studying preamble SINR outage \cite{Gharbieh2017,Jiang2018} without considering collision, or assumed the collision happens during preamble transmission for simplicity\cite{Gharbieh2018,Jiang2018a,Jiang2018b,Liu2020} and data transmission was ignored, whose model can not capture accurate average energy consumption and  average throughput during RA procedures.
	
	Due to the lack of theoretical analysis and comparison among 4/2-step and 4/2-step SDT RA schemes in existing works, we aim to address the following fundamental questions: 1) how to fairly model the 4/2-step, and 4/2-step SDT RA schemes with data transmission; 2) how to capture the overall packet transmission success probability, average throughput and energy consumption of each scheme; 3) which scheme performs better in a specific device density scenario. To do so, we develop a novel spatio-temporal mathematical framework to analyse the overall packet transmission success probability, average energy consumption, and average throughput under each RA scheme. The main contributions of this paper can be summarized in the following points:
	
	\begin{itemize}
		\item We present a tractable spatio-temporal mathematical framework to analyze 4/2-step, and 4/2-step SDT RA schemes based on stochastic geometry and probability theory, in which the devices are modelled as independent Poisson point process (PPP) in the spatial domain, and the new arrival packets of each device are modelled by independent Poisson arrival process in the time domain.
		\item  We jointly model preamble detection, PUSCH decoding, and data transmission procedures, which are general and can be extended to analyze different RA schemes. We derive the exact expressions for the preamble detection probability and the PUSCH decoding probabilities taking into account SINR outage and collision, and data transmission success probability. We also derive the overall packet transmission success probability, average throughput, and average energy consumption in each RACH attempt.

		\item We develop a realistic simulation framework to capture the preamble detection, PUSCH decoding, data transmission, where overall packet transmission success probability, average throughput, and average energy consumption for packet transmission are verified. Our results show that 2-step SDT RA scheme provides the highest overall packet transmission success probability, and lowest energy consumption. However, the performance gain compared with other RA schemes decreases with the increase of device intensity.
		
	\end{itemize}
	The rest of the paper is organized as follows. Section II presents the system model. Sections III derives the transmission success probability of each message, overall packet transmission success probability in each RACH attempt, and average throughput. Section IV derives the average energy consumption based on transmission status. Section V provides numerical results. Finally, Section VI concludes the paper.

	\section{System Model}
	We consider uplink transmission for the cellular network consisting of a single BS and multiple IoT devices\footnote{The Physical Random Access Channel (PRACH) root sequence planning is applied to mitigate the false alarm ratio of preamble detection among neighbouring BSs, thus neighbouring BSs have different preamble sets\cite{Ahmadi2019}. Therefore, we focus on studying the mathematical framework for Random Access Channel (RACH) in a single cell.\label{root}}, which are spatially distributed in $ \mathbb{R}^2 $ following
	independent homogeneous PPP ${\mathrm{\Phi_D}}$ with intensity $\mathrm{\lambda_D}$, and are assumed to be static all time once they are deployed.

	\subsection{Network and Traffic Model}
	
	We consider a flat Rayleigh fading channel, where the channel  between two generic locations $ x, y \in \mathbb{R}^2 $ is assumed to follow $ h(x,y) \sim \mathcal{CN}(0, 1)$. All the channel gains are independent of each other, independent of the spatial locations, and identically distributed (i.i.d.). For the brevity of exposition, the spatial indices $ (x, y) $ are dropped. We consider a standard power-law path-loss model with attenuation $ r^{-\alpha} $, where $ r $ is the propagation distance from the device to BS, and $ \alpha $ is the path-loss exponent. An ideal full path-loss inversion power control is assumed at all devices to solve the ``near-far'' problem based on downlink transmission path-loss estimation, where each device compensates for its path-loss to keep the average received signal power at the BS equal to the same threshold $ \rho $. We also assume that none of the devices suffer from truncation outage, which means the maximum transmit power of the device is large enough to compensate for uplink path-loss\cite{Jiang2018}.
	
	We model the new arrived packets $ N_\mathrm{new}^{m} $ in the $ m $th RACH attempt at each device using independent Poisson arrival process $ \Lambda_\mathrm{new}^{m} $ with the intensity $ \mu_\mathrm{new}^{m} $. The packets of each device line in a queue to be transmitted, and are determined by the newly arrived and undelivered packets. First Come First Serve (FCFS) scheduling scheme is applied by placing the newly arrived packets at the end of the queue. Without loss of generality, we assume the buffer size of each device is large enough, and an infinite amount of RACH attempts is assumed, where no packet is dropped until the packet is successfully received by the BS.

	\subsection{Contention-Based Random Access Schemes}

	We present the detailed procedures of the existing 4/2 RA schemes, and the newly proposed 4/2 SDT RA schemes in this section, respectively.

	\subsubsection{4-step and 4-step SDT Random Access}
	The 4-step RA scheme is shown in Fig.~\ref{fig:Procedures of each scheme}(a). In step 1, each device randomly selects a preamble generated by Zadoff-chu (ZC) sequence cyclic shift, and transmits as Msg1 on the RA subframe, which implicitly specifies the RA-radio network temporary identifier (RA-RNTI). In step 2, the BS responds to the device with Msg2 (i.e., random access response (RAR)) containing a temporary cell RNTI (TC-RNTI), timing advance(TA), and PUSCH resource granted for Msg3 transmission under the condition that BS successfully detects the preamble. If not, the device reattempts in the next RACH opportunity. In step 3, the device transmits Msg3 (i.e., RRC connection resume request) including a device identity. If multiple devices select the same preamble and RA subframe in Msg1, they receive the same Msg2, and transmit their own Msg3 on the same PUSCH resource resulting in a collision. If the BS successfully decodes one specific Msg3 among colliding devices, in step 4, the BS sends Msg4 (i.e., RRC connection resume) with an echo of the identity transmitted in Msg3 by the device. Those devices with matched identity succeed in the RA procedure and enter into the RRC Connected state for data transmission with HARQ, and all failed devices have to reattempt in the next RACH opportunity. After successful data transmission, the device receives the RRC release with suspend from the BS, and goes back to the RRC Inactive state. We only model data HARQ to make the comparison results more intuitive, and HARQ for both Msg3 and Msg4 is an easy extension.
	
	\begin{figure}[htbp!]
		\centering
		\includegraphics[scale=0.8]{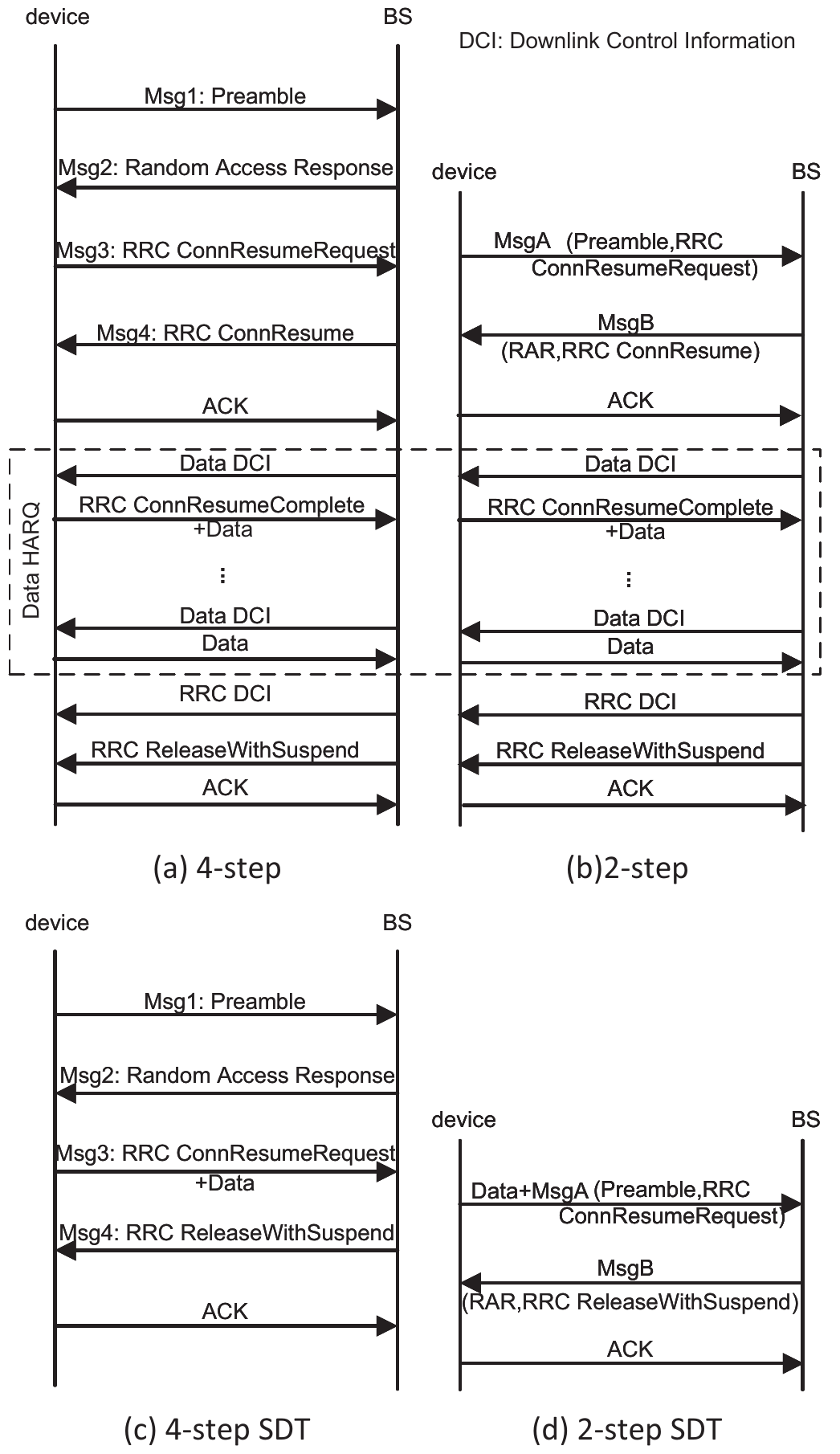}
		\caption{Procedures of each scheme.}
		\label{fig:Procedures of each scheme}
	\end{figure}

	The 4-step SDT RA scheme is illustrated in Fig.~\ref{fig:Procedures of each scheme}(c), which enables devices to transmit data with Msg3 without entering into RRC Connected state\cite{Hoglund2018,Thota2019}. Since the BS would have no knowledge whether the device has a small packet to transmit or not. Even if devices have deterministic traffic and the traffic pattern are predictable, the device identity is not known to the BS until Msg3. Therefore, the device is required to indicate its wish to use 4-step SDT RA scheme to the BS in Msg1 by randomly selecting one of the special PRACH preambles, which have been dedicated to 4-step SDT by the BS in system information. Once the preamble is detected successfully, the BS grants a larger PUSCH resource for both Msg3 and data transmission. Here, we assume that the packet size is smaller than the maximum transport block size (TBS), hence the data can be transmitted together with Msg3 in one slot. The device, whose Msg3 and data are successfully decoded among colliding devices, receives Msg4 (i.e., RRC release with suspend). Since the data is transmitted together with Msg3, no data HARQ needs to be considered in 4-step SDT RA scheme.
	
	\subsubsection{2-step and 2-step SDT Random Access}
	The 2-step and 2-step SDT RA schemes are illustrated in Fig.~\ref{fig:Procedures of each scheme}(b)(d), where data is either transmitted after MsgB in 2-step RA scheme in RRC Connected state, or together with MsgA in 2-step SDT RA scheme in RRC Inactive state, respectively. The MsgA consists of two parts, that is, preamble and PUSCH (corresponding to Msg3 in 4-step RA scheme), which are transmitted separately over time with independent channel. Unlike 4-step and 4-step SDT RA schemes, where the PUSCH resources are granted via Msg2, each preamble is mapped to a PUSCH in advance in 2-step and 2-step SDT RA schemes. If multiple preambles are mapped to a PUSCH in the same time-frequency resource, the PUSCH transmissions of the preambles overlap in time and frequency, which increases the probability of PUSCH decoding failure. Alternatively, a single preamble can be mapped to a unique PUSCH time-frequency resource, which reduces the probability of PUSCH decoding failure due to collision, but significantly increases the physical-layer overhead in the uplink. Here, we only consider the unique mapping relationship between preamble and PUSCH as the baseline. It is noted that if the preamble is detected but none of PUSCH transmission is decoded successfully among the colliding devices in 2-step and 2-step SDT RA schemes, the BS sends a fallback MsgB with an UL grant for Msg3 transmission, and all the colliding devices fallback to the 4-step RA scheme.

	\subsection{Random Access Model}
	To capture the characteristics of uplink transmission with different RA schemes, we jointly model the preamble detection, PUSCH decoding, and data transmission, which are based on ZC sequence characteristic, Power Delay Profile (PDP), SINR, and Block Error Rate (BLER), respectively.

	\subsubsection{Zadoff-chu Sequence Characteristic}
	 As we mentioned earlier, the preamble is generated from cyclic shift of the ZC sequence defined as 
	 	\begin{equation}
	 	z_{r}\left[ k\right]\triangleq \mathrm{exp}\left[ -j\pi rk(k+1)/N_{\mathrm{ZC}}\right],
	 	\end{equation}
	 	where $ k $ is the sequence index, $ N_{\mathrm{ZC}} $ denotes the sequence length, and $ r  $ is the root number broadcasted to devices in the system information, which is different among the neighbouring BSs\textsuperscript{\ref{root}}.
	
	The ZC sequences have an ideal cyclic auto-correlation property, which means the magnitude of the cyclic correlation with a circularly shifted version of itself becomes a scaled delta function as
		\begin{equation}
		\left| c_{rr}\left[  \tau \right] \right| =\left|\sum_{k=0}^{N_{\mathrm{ZC}}-1} z_{r}\left[ k\right] z_{r}^{*}\left[ k+\tau\right] \right|  =N_{\mathrm{ZC}}\delta\left[ \tau\right],
		\label{eq:preamble_corre}
		\end{equation}
		where  $ c_{rr}\left[  \tau \right] $ is the discrete cyclic auto-correlation function of $ z_{r}\left[ k\right] $ at lag $ \tau $ and $ [\cdot]^{*} $ denotes the complex conjugate. From this property, we can observe how much the received sequences are shifted, compared to the reference ZC sequence.

	In principle, multiple preambles can be generated from a ZC sequence by cyclically shifting the sequence by a factor of cyclic shift value $ \mathrm{N}_{\mathrm{CS}} $. Thus, the $ i $th preamble can be represented as  
		\begin{equation}
		z_{r}^{i}\left[ k\right]=z_{r}\left[ \left( k+i \mathrm{N}_{\mathrm{CS}}\right) \bmod  N_{\mathrm{ZC}} \right].
		\label{eq:reference preamble}
		\end{equation}

	\subsubsection{Power Delay Profile}
	
	PDP is utilized to model the preamble detection (i.e., Msg1 or MsgA) at the BS, which is the periodic correlation of the received preamble as a function of time. The BS detects preamble transmission via PDP, which indicates whether devices are requesting resources for uplink transmission. We assume that there is no interference and collision during preamble transmission by considering the BS having enough resolution to distinguish the propagation delays among any two or more users \cite{Jang2017}. In particular, in Fig.~\ref{fig:PDP}, we show that the PDP peak values of devices choosing the same preamble (UE1+UE2 at $\tau_{12}$ in the figure) are indeed separable when the BS is able to achieve a higher resolution (UE1 and UE2 at $\tau_1$ and $\tau_2$).
	
	\begin{figure}[htbp!]
	\centering
	\includegraphics[scale=1.5]{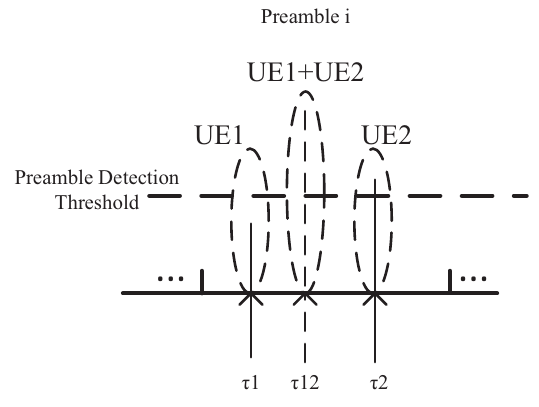}
	\caption{Power Delay Profile.}
	\label{fig:PDP}
\end{figure}

	Specifically, the BS extracts the received preamble within specific time/frequency resources through time-domain sampling and frequency-tone extraction. The received preamble from a typical device can be written as \cite{Ta2019,Shirvanimoghaddam2017,Lin2016,Wang2015,Liang2017}
		\begin{equation}
		{y_{r}^{i}\left[ k\right]} =  {\sqrt {\rho }{h_0} } z_{r}^{i}\left[ k+\tau_{0}\right] + n_0,
		\label{eq:received preamble}
		\end{equation}
		where $ \tau_{0} $ is the sequence shift caused by the propagation delay from a typical device to the BS, $ \rho $ is the power control threshold, $ h_{0} $ is the channel between the BS and a typical device, and $ n_0 $ denotes complex Gaussian noise with zero mean and variance $\sigma_{n}^2$.
	
	The BS computes PDP of a typical device via time-domain correlation between the received preamble \eqref{eq:received preamble} and the local reference preamble sequence \eqref{eq:reference preamble}. Therefore, we formulate the PDP of a typical device in the $ m $th RACH attempt as
	\begin{equation}
	\begin{aligned}
	{\mathrm{PDP}^m\left[ \tau\right] } &= {\left|\sum_{k=0}^{N_{\rm ZC}-1}{y_{r}^{i}\left[ k\right]}z_{r}^{i}\left[ k+\tau\right]^*\right|^2}\\&=\left| {\sum\limits_{k = 0}^{{{N}_{\rm{ZC-1}}}} \left( {\sqrt {\rho }{h_0} } z_{r}^{i}\left[ k+\tau_{0}\right] + n_0\right) {z_{r}^{i}\left[ k+\tau\right]}^ * } \right|^2.
	\end{aligned}
	\label{eq:power delay profile}
	\end{equation}

	\subsubsection{Signal to Noise plus Interference Ratio}
	SINR is utilized to model the PUSCH decoding (i.e., Msg3 or MsgA) when the preamble is successfully detected at the BS. As we mentioned earlier, each device transmits a randomly chosen preamble to the BS, and thus devices choosing same preamble in the same RA subframe cause the intra-cell interference in PUSCH transmission. We formulate $ \mathrm{SINR} $ of the PUSCH transmission in $ m $th RACH attempt as \cite{Jiang2018,Gharbieh2017}
	\begin{equation}
	{{\mathrm{SINR}}^{m}} = {\rho | h_{0}| ^2}/\left( {\mathcal{I}_{\mathrm{intra}}^{m}+\sigma_{n}^2}\right) 
	\label{eq:SINR},
	\end{equation}
	where
	\begin{equation}
	{{\mathcal {I}}^{m}_{\mathrm{ intra}}}=\sum \limits _{{j} \in {\mathcal{ Z}}_{\mathrm{ in}} }{{{\mathsf{1}}}_{\{ {N^{m}_{{\rm New}_{j}}}+{N^{m}_{{\rm Cum}_{j}}} > 0\} }} {\rho |{h_{j}}|^2}.
	\label{eq:Intra interference}
	\end{equation}
	
	In \eqref{eq:SINR}, $ h_{0} $ is the channel from a typical device to the BS, and $ \mathcal{I_{\mathrm{intra}}} $ is the aggregated intra-cell interference in the $ m $th RACH attempt.
	In \eqref{eq:Intra interference}, $ {\mathcal{ Z}}_{\mathrm{ in}} $ is the set of intra-cell interfering devices, $ {N^{m}_{{\rm New}_{j}}} $ is the number of new arrived packets in the $ m $th RACH attempt of $ j $th interfering UE, $ {N^{m}_{{\rm Cum}_{j}}} $ is the number of accumulated packets in the $ m $th RACH attempt of $ j $th interfering device. $ \mathsf{1}_{\left\lbrace . \right\rbrace } $is the indicator function that takes the value 1 if the statement $ \mathsf{1}_{\left\lbrace . \right\rbrace } $ is true, and zero otherwise. Whether the device generates interference is determined by the fact that if the condition $ {{\mathsf{1}}}_{\{ {N^{m}_{{\rm New}_{j}}}+{N^{m}_{{\rm Cum}_{j}}} > 0\} } $ has been satisfied. This means that the device is able to generate interference only when its buffer is non-empty.

	\subsubsection{Block Error Rate}
	For data transmission after the successful RA procedure (i.e., 4/2-step RA schemes), we assume the BLER of each data transmission is $ B $ (e.g., $ B=0.1 $). This is because the BS performs link adaption to adaptively change the modulation and coding scheme, which can guarantee the BLER target \cite{Bruno2014, Yu2017}. If the data is transmitted during the RA procedure (i.e., 4/2-step SDT RA schemes), we do not consider data transmission separately for fairness.

	\section{Transmission Success Probability and average throughput}
	In this section, we analyze the preamble detection, PUSCH decoding, data transmission, overall packet transmission success probability, and average throughput. As we mentioned earlier, the devices are spatially distributed in $ \mathbb{R}^2 $ following independent homogeneous PPP ${\mathrm{\Phi_D}}$ with intensity $\mathrm{\lambda_D}$. We assume that the BS has an available preamble pool with the number of non-dedicated preambles $ \xi $, known by the devices. Each preamble has an equal probability $ 1/\xi $ to be chosen by the device, hence the average number of the devices using the same preamble $ \lambda_{\mathrm{Dp}} $ is 
		\begin{equation}
		\lambda_{\mathrm{Dp}} = \lambda_{\mathrm{D}}/\xi,
		\end{equation}
		where $ \lambda_{\mathrm{D}} $ is the device intensity in the cell.
	
	The device generates interference only when it has data to transmit in the $m$th RACH attempt, thus we define the non-empty buffer probability of a typical device as $ \mathcal{T}^{m}=\mathbb{P}\left[ {N^{m}_{{\rm New}}}+{N^{m}_{{\rm Cum}}} > 0 \right] $, where $N^{m}_{{\rm New}}$ is the number of new arrived packets in the $m$th RACH attempt, and $N^{m}_{{\rm Cum}}$ is the number of accumulated packets in the $m$th RACH attempt.
	
	As the new packets arrival at each device $ N_\mathrm{new}^{m} $ in the $m$th RACH attempt is modelled by independent Poisson process $ \Lambda_\mathrm{new}^{m} $ with the intensity $ \mu_\mathrm{new}^{m} $, the packets departure can be treated as an approximated thinning process (i.e., the thinning factor is a function relating to the overall transmission success probability $ \mathcal{P}^{m} $, and the non-empty buffer probability $\mathcal{T}^{m}$) of the arrived packets. Therefore, after this thinning process in a specific RACH attempt, the least packets (i.e. the accumulated packets $ N_{\mathrm{Cum}}^{m} $) number at each device can be approximated as Poisson distribution $ \Lambda_\mathrm{cum}^{m} $ with the intensity $ \mu_\mathrm{cum}^{m} $.
	
	Similar as \cite{Jiang2018}, the intensity of accumulated
		packets $\mu_\mathrm{cum}^{m} (m > 1)$ in the $m$th RACH attempt is derived as
		\begin{equation}
		\mu^m_{\mathrm{Cum}}=\mu^{m-1}_{\mathrm{New}}+\mu^{m-1}_{\mathrm{Cum}}-\mathcal{P}^{m-1}\mathcal{T}^{m-1},
		\end{equation}
		where $\mu^{m-1}_{\mathrm{New}}$ and $\mu^{m-1}_{\mathrm{Cum}}$ are the intensity of new packets and accumulated packets in the $m-1$ RACH attempt. $\mathcal{P}^{m-1}$ and $\mathcal{T}^{m-1}$ are the overall transmission success probability and non-empty buffer probability in the Poisson approximation. Thus, the non-empty probability of each device in the $m$th RACH attempt is derived as 
		\begin{equation}
		\mathcal{T}^{m}=1-e^{-\mu^{m}_{\mathrm{New}}-\mu^{m}_{\mathrm{Cum}}}.
		\label{eq:non-empty probability}
		\end{equation}
	
	Therefore, the intensity of interfering devices is $\lambda_{Dp}\mathcal{T}^{m}$, and the Probability Mass Function (PMF) of the number of interfering devices in the $m$th RACH attempt is derived as 
		\begin{equation}
		\mathbb{P}\left[N=n\right]={e^{-\lambda_{\mathrm{Dp}}\mathcal{T}^{m}}\left(\lambda_{\mathrm{Dp}}\mathcal{T}^{m} \right)^{n}}/{n!},
		\label{eq:Interfering number}
		\end{equation}
		where $N$ is the number of interfering devices.

	\subsection{Preamble Detection in Msg1 and MsgA}
	The preamble detection is performed at step 1 of each RA scheme by calculating the PDP, where the peak values of devices choosing the same preamble are separate due to different propagation delays between each device and BS. It is important to note the BS cannot detect the collision during preamble detection in step 1 of RA schemes, because the multiple peak values may be also caused by the multipath effect\cite{anton2014machine, Sesia2011, Dahlman2013}\footnote{When the cell size is more than twice the distance corresponding to the maximum delay spread, the BS may be able to differentiate two devices selecting the same preamble.}.
	
	For the brevity of exposition, we define the event with $ n $ interfering devices as $ \mathcal{A}= \left\lbrace N=n\right\rbrace  $. Considering that the preamble detection is successful when at least one PDP peak values among $ n+1 $ colliding devices is above the threshold $ \lambda _{{\rm{th}}} $, the preamble detection success probability of a typical device in $ m $th RACH attempt conditioning on $ n $ interfering devices is defined as
		\begin{equation}
		{\cal P}_{{\rm{pre}|\mathcal{A}}}^m =1 - \prod\limits_{l = 1}^{n + 1} {\mathbb{P}[\mathrm{PDP}_{l}^{m}\left[ \tau_{l}\right] < {\lambda _{{\rm{th}}}}|\mathcal{A}]},
		\label{eq:preamble}
		\end{equation}
		which is characterized in the following lemma.
	
	\begin{lemma}
		The preamble detection success probability of a typical device in the $ m $th RACH attempt conditioning on $ n $ interfering devices is derived as
		\begin{equation}
		{\cal P}_{{\rm{pre}|\mathcal{A}}}^m =1 - \left[1-\mathrm{ exp}\left( -\dfrac{\lambda _{{\rm{th}}}}{\rho N_{\rm{ZC}}^2 + \sigma_{n}^2 N_{\rm{ZC}}}\right) \right]^{n+1},
		\label{eq:preamble transmission success probability}
		\end{equation}
		where $ N_{\rm{ZC}} $ is the length of the preamble sequence, $ \sigma_{\mathrm{n}}^2 $ is the noise power, $ \rho $ is the average received power at the BS, and $ \lambda _{{\rm{th}}} $ is the threshold for preamble detection.

	\end{lemma}
	
	\begin{proof}
		See Appendix A.
	\end{proof}

	We then derive the overall preamble detection success probability of a typical device in the following \textbf{Corollary 1}.
	
	\begin{corollary}
		The overall preamble detection success probability of a typical device in $ m $th time slot is derived as
		\begin{equation}
		\begin{split}
		\mathcal{P}^{m}_{\mathrm{pre}} &=  \sum_{n=0}^{\infty}\mathbb{P}\left[N=n\right] {\cal P}_{{\rm{pre}|\mathcal{A}}}^m= \sum_{n=0}^{\infty} \dfrac{e^{-\lambda_{\mathrm{Dp}}\mathcal{T}^{m}}\left(\lambda_{\mathrm{Dp}}\mathcal{T}^{m} \right)^{n}}{n!}\times\\& \left\lbrace 1 - \left[1-\mathrm{ exp}\left( -\dfrac{\lambda _{{\rm{th}}}}{\rho N_{\rm{ZC}}^2 + \sigma_{n}^2 N_{\rm{ZC}}}\right) \right]^{n+1}\right\rbrace.
		\end{split}
		\label{eq:overall_preamble}
		\end{equation}
		
	\end{corollary}
	
		\begin{proof}
		See the proof of \textbf{Lemma 1}.
	\end{proof}

	\subsection{PUSCH Decoding in Msg3 and MsgA}
	The PUSCH decoding is performed either in step 3 of 4-step and 4-step SDT RA schemes, or step 1 of the 2-step and 2-step SDT RA schemes. In order to analyse the collision happens in Msg3 or MsgA, we derive the PUSCH decoding success probability of a typical device in $ m $th RACH attempt conditioning on $ n $ interfering devices and preamble detection success. For the brevity of exposition, we define the preamble detection success of a typical device with $ n $ interfering devices as event $ \mathcal{B} =\left\lbrace\prod\limits_{l = 1}^{n + 1} {{{\mathsf{1}}}_{\{ \mathrm{PDP}_{l}^{m}< {\lambda _{{\rm{th}}}}\} }= 0}, \mathcal{A} \right\rbrace$, where ${{{\mathsf{1}}}_{\{ \mathrm{PDP}_{l}^{m}< {\lambda _{{\rm{th}}}}\} }}$ takes value 1 if the PDP peak value of device $l$ is below the threshold ${\lambda _{{\rm{th}}}}$.

	For the basic receiver, the BS has no multi-user detection capabilities, which means all transmissions are failed when a collision takes place \cite{Jiang2018a}. However, for an advanced receiver, a specific signal can still be decoded  when multiple signals are received at different powers, this is the so called capturing capability. This phenomenon occurs when the strongest signal power received from a typical device is sufficiently large. In other words, the capture effect occurs when the SINR\footnote{In \cite{Kim2017,Yue1991}, the signal-to-interference-ratio (SIR) is used, however, we use SINR as \cite{Zanella2012}} of the strongest signal is larger than a specific threshold. As we have assumed perfect power control to compensate for the path-loss, the SINR in the advanced receiver actually depends on the small-scale Rayleigh fading. Therefore, the PUSCH decoding success probability of a typical device in $ m $th RACH attempt conditioning on $ n $ interfering devices and preamble detection success is  defined as
		\begin{equation}
		\begin{aligned}
		&{\cal P}_{{\mathrm{pus}|\mathcal{B}}}^{m} =  {\mathbb{P}\left[ {\rm{SINR}}_{o} > {\gamma _{{\rm{th}}}},|h_{o}|^2>|h_{j}|^2  |\mathcal{B}\right] }\\&= {\mathbb{P}[{\mathrm{SINR}}_{o} > {\gamma _{{\mathrm{th}}}}||h_{o}|^2>|h_{j}|^2 ,\mathcal{B}]}\mathbb{P}[{|h_{o}|^2>|h_{j}|^2|\mathcal{B}] },
		\label{eq:Msg3 capture effect model}
		\end{aligned}
		\end{equation}
		where ${j} \in {\mathcal{ Z}}_{\mathrm{ in}}$ represents the set of intra-cell interfering devices in \eqref{eq:Intra interference}, and the PUSCH decoding success probability is characterized in the following lemma.
	
	\begin{lemma}
		The PUSCH decoding success probability conditioning on $ n $ interfering devices and preamble detection success is derived as
		\begin{equation}
		{\cal P}_{{\mathrm{pus}|\mathcal{B}}}^{m} = \dfrac{\sum_{k=1}^{n+1}\tbinom{n+1}{k}(-1)^{k+1}\mathrm{ exp}\left\lbrace \dfrac{-k\gamma_{\rm{th}}\sigma_{\mathrm{n}}^{2}}{\rho}\right\rbrace}{ {\left( n+1\right) (1+\gamma _{th})^{n}}},
		\label{eq:capture final result}
		\end{equation}
		where $ \sigma_{\mathrm{n}}^{2} $ is the average noise power, $ \rho $ is the average received power at the BS, and $ \gamma _{{\rm{th}}} $ is the SINR threshold for PUSCH decoding.
	\end{lemma}
	
	\begin{proof}
		See Appendix B.
	\end{proof}
	
		We then derive the overall PUSCH decoding success probability of a typical device in $ m $th time slot in the following \textbf{Corollary 2}.
	\begin{corollary}
		The overall PUSCH decoding success probability of a typical device in $ m $th time slot is derived as
		\begin{equation}
		\begin{split}
		&\mathcal{P}^{m}_{\mathrm{pus}}=\\&  \sum_{n=0}^{\infty}\mathbb{P}\left[N=n \right] {\cal P}_{{\rm{pre}|\mathcal{A}}}^m{\cal P}_{{\rm{pus}|\mathcal{B}}}^{m}=\sum_{n=0}^{\infty} \dfrac{e^{-\lambda_{\mathrm{Dp}}\mathcal{T}^{m}}\left(\lambda_{\mathrm{Dp}}\mathcal{T}^{m} \right)^{n}}{n!}\times\\& \left\lbrace 1 - \left[1-\mathrm{ exp}\left( -\dfrac{\lambda _{{\rm{th}}}}{\rho N_{\rm{ZC}}^2 + \sigma_{n}^2 N_{\rm{ZC}}}\right) \right]^{n+1}\right\rbrace\times\\& \dfrac{\sum_{k=1}^{n+1}\tbinom{n+1}{k}(-1)^{k+1}\mathrm{ exp}\left\lbrace \dfrac{-k\gamma_{\rm{th}}\sigma_{\mathrm{n}}^{2}}{\rho}\right\rbrace }{\left( n+1\right) (1+\gamma _{th})^{n}}.
		\end{split}
		\label{eq:overall_msg3_u}
		\end{equation}
	\end{corollary}
	\begin{proof}
		See the proofs of \textbf{Lemma 1} and \textbf{Lemma 2}.
	\end{proof}

	\subsection{Data Transmission after RACH}
	For the brevity of exposition, we define the preamble detection success, and PUSCH decoding success with $ n $ interfering devices as event $ \mathcal{C}=\left\lbrace{\rm{SINR}}_{o} > {\gamma _{{\rm{th}}}},|h_{o}|^2>|h_{j}|^2, \mathcal{B} \right\rbrace $. We assume the BLER of each data transmission is $ B$\footnote{As the overall successful transmission probability is mainly limited by the successful RACH probability, we consider basic independent decoding process in each data transmission without chase-combing or incremental redundancy mechanism to make the comparison results more intuitive.}, hence the data transmission success probability of a typical device in $ m $th RACH attempt conditioning on event $ \mathcal{C} $ is derived as
	\begin{equation}
	{\cal P}_{{\mathrm{data}| \mathcal{C} }}^{m} = 1-B^{\mathrm{K}},
	\label{eq:data probability}
	\end{equation}
	where $ \mathrm{K} $ is the maximum data HARQ transmission times.

	\subsection{Overall Packet Transmission Success Probability}
	\subsubsection{4-step Random Access}
	Since the data is transmitted after successful 4-step RA procedure, the overall packet transmission success only occurs when preamble detection, PUSCH decoding, and data transmission are all successful. Hence, the overall packet transmission success probability of 4-step RA scheme in $ m $th RACH attempt is derived as
	\begin{equation}
	\begin{split}
	&\mathcal{P}^{m}_{\mathrm{4step}}=\sum_{n=0}^{\infty}\mathbb{P}\left[N=n\right]{\cal P}_{{\rm{pre}|\mathcal{A}}}^m{\cal P}_{{\rm{pus}|\mathcal{B}}}^{m}{\cal P}_{\rm{data}|\mathcal{C}}\\&=\sum_{n=0}^{\infty}\left[ {e^{-\lambda_{\mathrm{Dp}}\mathcal{T}^{m}_{\mathrm{4step}}}\left(\lambda_{\mathrm{Dp}}\mathcal{T}^{m}_{\mathrm{4step}} \right)^{n}}/{n!}\right]\times\\& \left\lbrace 1 - \left[1-\mathrm{ exp}\left( -\dfrac{\lambda _{{\rm{th}}}}{\rho N_{\rm{ZC}}^2 + \sigma_{n}^2 N_{\rm{ZC}}}\right)\right]^{n+1}\right\rbrace\times\\&
	\dfrac{\sum_{k=1}^{n+1}\tbinom{n+1}{k}(-1)^{k+1}\mathrm{ exp}\left\lbrace \dfrac{-k\gamma_{\rm{th}}\sigma_{\mathrm{n}}^{2}}{\rho}\right\rbrace\left( 1-B^{\mathrm{K}}\right) }{\left( n+1\right) (1+\gamma _{th})^{n}},
	\end{split}
	\label{eq:four step overall success probability}
	\end{equation}
	where $\mathcal{T}^{m}_{\mathrm{4step}}$ is non-empty buffer given in \eqref{eq:non-empty probability}.
	\subsubsection{4-step SDT Random Access}
	The overall packet transmission success occurs when both preamble detection and PUSCH decoding are successful. Hence, the overall packet transmission success probability of 4-step SDT in $ m $th RACH attempt is derived as
	\begin{equation}
	\begin{aligned}
	&\mathcal{P}^{m}_{\mathrm{4stepSDT}}=\sum_{n=0}^{\infty}\mathbb{P}\left[N=n\right]{\cal P}_{{\rm{pre}|\mathcal{A}}}^m{\cal P}_{{\rm{pus}|\mathcal{B}}}^m\\&=\sum_{n=0}^{\infty}\left[ {e^{-\lambda_{\mathrm{Dp}}\mathcal{T}^{m}_{\mathrm{4stepSDT}}}\left(\lambda_{\mathrm{Dp}}\mathcal{T}^{m}_{\mathrm{4stepSDT}} \right)^{n}}/{n!}\right]\times\\& \left\lbrace 1 - \left[1-\mathrm{ exp}\left( -\dfrac{\lambda _{{\rm{th}}}}{\rho N_{\rm{ZC}}^2 + \sigma_{n}^2 N_{\rm{ZC}}}\right)\right]^{n+1}\right\rbrace\times\\& \dfrac{\sum_{k=1}^{n+1}\tbinom{n+1}{k}(-1)^{k+1}\mathrm{ exp}\left\lbrace \dfrac{-k\gamma_{\rm{th}}\sigma_{\mathrm{n}}^{2}}{\rho}\right\rbrace}{ {\left( n+1\right) (1+\gamma _{th})^{n}}}.
	\end{aligned}
	\label{eq:EDT overall success probability}
	\end{equation}

	\subsubsection{2-step Random Access}
	Remind that if the BS detects the preamble but fails to decode any PUSCH signal among $n+1$ colliding devices in MsgA, the fallback mechanism in 2-step RA scheme allows the devices to transmit Msg3 following the 4-step RA scheme. Therefore, we define the fallback probability, and PUSCH decoding success probability after the fallback as
		\begin{equation}
		\begin{split}
		\mathcal{P}^{m}_{\mathrm{fb}} =  \sum_{n=0}^{\infty}\mathbb{P}\left[N=n\right] {\cal P}_{{\rm{pre}|\mathcal{A}}}^m\left( 1-\left(n+1 \right) {\cal P}_{{\rm{pus}|\mathcal{B}}}^m\right), 
		\label{eq:overall_fb}
		\end{split}
		\end{equation}
		and
		\begin{equation}
		\begin{split}
		\mathcal{P}^{m}_{\mathrm{fb\_pus}} =&  \sum_{n=0}^{\infty}\mathbb{P}\left[N=n\right] {\cal P}_{{\rm{pre}|\mathcal{A}}}^m\left( 1-\left(n+1 \right) {\cal P}_{{\rm{pus}|\mathcal{B}}}^m\right){\cal P}_{{\rm{pus}|\mathcal{B}}}^m,
		\end{split}
		\label{eq:overall_fbmsg3}
		\end{equation}
		where term $\left( 1-\left(n+1 \right) {\cal P}_{{\rm{pus}|\mathcal{B}}}^m\right)$ is the probability that the BS fails to decode any MsgA PUSCH signal among $(n+1)$ colliding devices.
	
	Since the packet can be successfully transmitted either after 2-step RA procedure with successful preamble detection and PUSCH decoding, or after 4-step RA procedure with successful fallback PUSCH decoding, the overall packet transmission success probability is derived as
	\begin{equation}
	\begin{aligned}
	&{{\cal P}^m_{\mathrm{2step}}}=\sum_{n=0}^{\infty}\mathbb{P}\left[N=n\right]{\cal P}_{{\rm{pre}|\mathcal{A}}}^m{\cal P}_{{\rm{pus}|\mathcal{B}}}^m{\cal P}_{{\rm{data}|\mathcal{C}}}^m+\\&\sum_{n=0}^{\infty}\mathbb{P}\left[N=n\right]{\cal P}_{{\rm{pre}|\mathcal{A}}}^m\left(1-(n+1){\cal P}_{{\rm{pus}|\mathcal{B}}}^m\right){\cal P}_{{\rm{pus}|\mathcal{B}}}^m{\cal P}_{\rm{data}|\mathcal{C}}\\&=\sum_{n=0}^{\infty}\left[ {e^{-\lambda_{\mathrm{Dp}}\mathcal{T}^{m}_{\mathrm{2step}}}\left(\lambda_{\mathrm{Dp}}\mathcal{T}^{m}_{\mathrm{2step}} \right)^{n}}/{n!}\right]\times\\& \left\lbrace 1 - \left[1-\mathrm{ exp}\left( -\dfrac{\lambda _{{\rm{th}}}}{\rho N_{\rm{ZC}}^2 + \sigma_{n}^2 N_{\rm{ZC}}}\right)\right]^{n+1}\right\rbrace\times\\& \dfrac{\sum_{k=1}^{n+1}\tbinom{n+1}{k}(-1)^{k+1}\mathrm{ exp}\left\lbrace \dfrac{-k\gamma_{\rm{th}}\sigma_{\mathrm{n}}^{2}}{\rho}\right\rbrace\left( 1-B^{\mathrm{K}}\right) }{\left( n+1\right) (1+\gamma _{th})^{n}}\times\\&\left\lbrace 2-\dfrac{\sum_{k=1}^{n+1}\tbinom{n+1}{k}(-1)^{k+1}\mathrm{ exp}\left\lbrace \dfrac{-k\gamma_{\rm{th}}\sigma_{\mathrm{n}}^{2}}{\rho}\right\rbrace}{(1+\gamma _{th})^{n}}\right\rbrace .
	\end{aligned}
	\label{eq:2stepB_overall}
	\end{equation}
	
	\subsubsection{2-step SDT Random Access}
	Similar as 2-step RA scheme, the overall packet transmission success probability of 2-step SDT RA scheme is derived as
	\begin{equation}
	\begin{aligned}
	&{{\cal P}^m_{\mathrm{2stepSDT}}} =\sum_{n=0}^{\infty}\mathbb{P}\left[N=n\right]{\cal P}_{{\rm{pre}|\mathcal{A}}}^m{\cal P}_{{\rm{pus}|\mathcal{B}}}^m+\\&\sum_{n=0}^{\infty}\mathbb{P}\left[N=n\right]{\cal P}_{{\rm{pre}|\mathcal{A}}}^m\left(1-(n+1){\cal P}_{{\rm{pus}|\mathcal{B}}}^m\right){\cal P}_{{\rm{pus}|\mathcal{B}}}^m{\cal P}_{\rm{data}|\mathcal{C}}\\&=\sum_{n=0}^{\infty}\left[ {e^{-\lambda_{\mathrm{Dp}}\mathcal{T}^{m}_{\mathrm{2stepSDT}}}\left(\lambda_{\mathrm{Dp}}\mathcal{T}^{m}_{\mathrm{2stepSDT}} \right)^{n}}/{n!}\right]\times\\& \left\lbrace 1 - \left[ 1-\mathrm{ exp}\left( -\dfrac{\lambda _{{\rm{th}}}}{\rho N_{\rm{ZC}}^2 + \sigma_{n}^2 N_{\rm{ZC}}}\right)\right]^{n+1}\right\rbrace\times\\& \dfrac{\sum_{k=1}^{n+1}\tbinom{n+1}{k}(-1)^{k+1}\mathrm{ exp}\left\lbrace \dfrac{-k\gamma_{\rm{th}}\sigma_{\mathrm{n}}^{2}}{\rho}\right\rbrace}{\left( n+1\right) (1+\gamma _{th})^{n}}\Vast\{ 1+\vast\{\\& 1-\dfrac{\sum_{k=1}^{n+1}\tbinom{n+1}{k}(-1)^{k+1}\mathrm{ exp}\left\lbrace \dfrac{-k\gamma_{\rm{th}}\sigma_{\mathrm{n}}^{2}}{\rho}\right\rbrace}{(1+\gamma _{th})^{n}}\vast\}\left( 1-B^{\mathrm{K}}\right)\Vast\}.
	\end{aligned}
	\label{eq:2stepA_overall}
	\end{equation}

		\subsection{Average Throughput}
We analyze the system throughput performance by assuming that all packets have the same packet size $ S $. As 2-step and 2-step SDT RA schemes integrate the fallback mechanism, the lengths of complete 2-step and 2-step SDT RACH procedure are the same as that of 4-step RA scheme. Therefore, we assume the length of 4-step RACH procedure as the period of system RACH opportunity $ \mathrm{T_{{RACH}}}$, and the throughput in the $ m $th RACH attempt is defined as
	\begin{equation}
	\begin{aligned}
	R^{m}={\mathcal{T}^{m}{\cal P}^mS}/{\mathrm{T_{{RACH}}}},
	\end{aligned}
	\label{eq:throughput}
	\end{equation}
	where $ {\cal T}^{m}{\cal P}^mS$ represents the total size of all successfully transmitted packet in $ m $th RACH attempt of each scheme,  which can be calculated based on \eqref{eq:non-empty probability}, \eqref{eq:four step overall success probability}, \eqref{eq:2stepB_overall}, and \eqref{eq:2stepA_overall}.
	
	\section{Average Energy Consumption}

	In this section, we analyze the average energy consumption of a typical device in $ m $th RACH attempt with 4/2-step, 4/2-step SDT RA schemes, respectively. Specifically, we derive the energy consumption of each message under either success or failure, and then calculate the average energy consumption based on the transmission success probability of each message in Section III.
	
	For ease of description, we first present definitions for general variables.  As illustrated in Fig.~\ref{fig:Timing relationship of 4-step and NR EDT RA schemes} and Table~\ref{notations}, $ \mathrm{T_{p}}$ is the preamble transmission time\cite{Dahlman2018}, $ \mathrm{T_{s}} $ is the slot time length, $  \mathrm{T_{d}} $ is the Physical Downlink Control Channel (PDCCH) decoding time, which contains the DCI, $ \mathrm{N_{RAR}} $ is the number of slots that RAR window occupies, $ \mathrm{N_{CRT}} $ is the number of slots that contention resolution timer (CRT) occupies, $\mathrm{N_{DCI}} $ is the number of slots between uplink data transmission and downlink DCI.  $ \mathrm{T_{K2}} $, $ \mathrm{T_{\triangle}} $ and $ \mathrm{T_{PUCCH}} $ are scheduling parameters defined in the standards\cite{Physicalcontrol, Physicaldata}. $ \mathrm{P_{s}} $ and $ \mathrm{P_{r}} $ are the power consumption when the device is in the sleep and receiving states separately, which are constants for all devices. As for the transmit power $ \mathrm{P_{t}} $, although the radiated power of each device depends on its distance to the BS due to the full-path power control, the power consumed by the amplifiers and the radio frequency (RF) hardware is actually much higher, almost regardless of the radiated power \cite{powersaving}. Therefore, the transmit power $ \mathrm{P_{t}} $ is also regarded as a constant.

	\begin{table}[htbp!]
		\caption{Notation Table}
		\begin{center}
			\begin{tabular}{|c|c|}
				\hline 
				\rowcolor{gray!25} \textbf{Notations}& \textbf{Physical means}  \\ 
				\hline
				$ \mathrm{T_{p}}$& The preamble transmission time  \\
				\hline
				\rowcolor{gray!25}$ \mathrm{T_{s}}$& The slot time length \\
				\hline
				$  \mathrm{T_{d}} $& The PDCCH decoding time \\
				\hline
				\rowcolor{gray!25}$  \mathrm{N_{RAR}} $& The number of slots that RAR window occupies \\
				\hline
				$  \mathrm{N_{CRT}} $& The number of slots that CRT occupies \\
				\hline
				$  \mathrm{N_{DCI}} $& The number of slots between data and downlink DCI \\
				\hline
				\rowcolor{gray!25}$  \mathrm{T_{K2}} $& The PUSCH scheduling parameter \\
				\hline
				$  \mathrm{T_{\triangle}} $& The specific PUSCH scheduling parameter for Msg3\\
				\hline
				\rowcolor{gray!25}$  \mathrm{T_{PUCCH}} $&The PUCCH scheduling parameter \\
				\hline
				$  \mathrm{P_{s}} $& The sleep power of a typical device \\
				\hline	
				\rowcolor{gray!25}$  \mathrm{P_{r}} $& The receiving power of a typical device \\
				\hline
				$  \mathrm{P_{t}} $& The transmit power of a typical device \\
				\hline		
			\end{tabular}
			\label{notations}
		\end{center}
	\end{table}

	\subsection{4-step Random Access}
	The energy consumption of 4-step RA scheme depends on the success or failure of preamble detection and PUSCH decoding.
	Therefore, we derive the energy consumption of each message with successful RA procedure, preamble detection failure, and PUSCH decoding failure, respectively.
	\begin{figure*}[htbp!]
		\centering
		\includegraphics[scale=0.8]{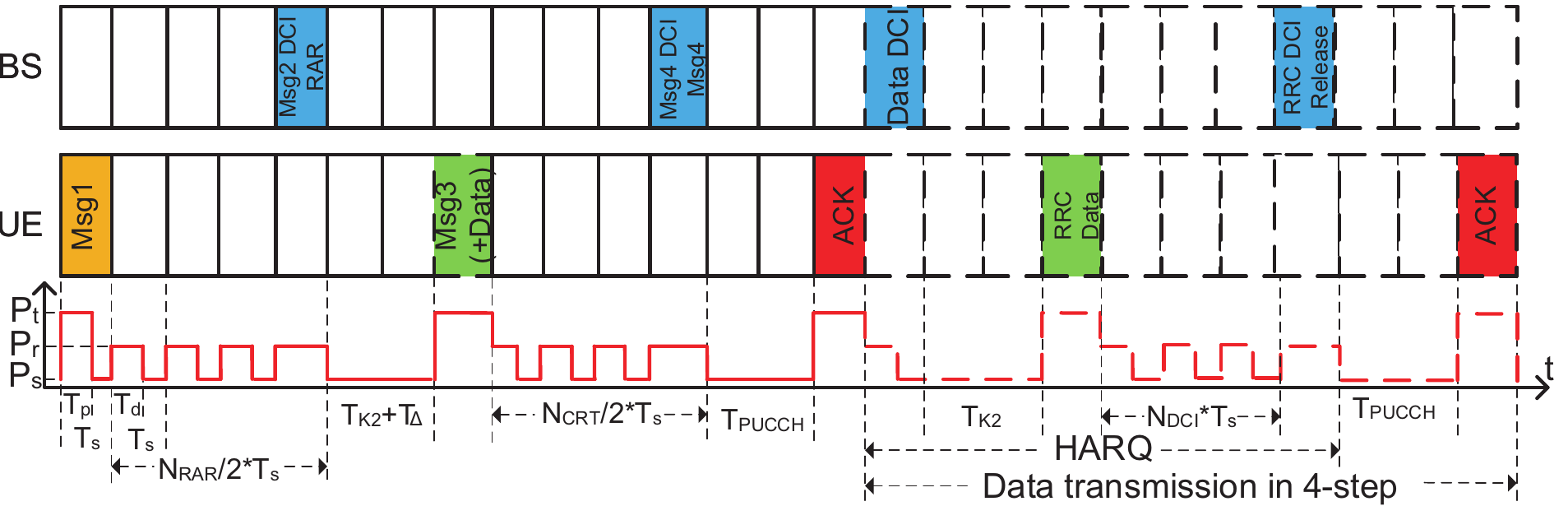}
		\caption{Timing relationship of successful 4-step and 4-step SDT RA procedure.}
		\label{fig:Timing relationship of 4-step and NR EDT RA schemes}
	\end{figure*}
	
	\subsubsection{Successful 4-step RA procedure} 
	The successful 4-step RA procedure timing relationship is illustrated in Fig.~\ref{fig:Timing relationship of 4-step and NR EDT RA schemes}, which occurs when both preamble detection and PUSCH decoding are successful.

	The device transmits Msg1 (i.e., preamble) in one slot $ \mathrm{T_{s}} $. Since the preamble transmission time only lasts for $ \mathrm{T_{p}} $, the device is in sleep state in the rest of time $ \mathrm{T_{s}-T_{p}} $. Therefore, the energy consumption of step 1 in 4-step RA scheme can be written as
	\begin{equation}
	E_{\mathrm{p}} = \mathrm{P_{t}T_{p}}+\mathrm{P_{s}\Big( T_{s}-T_{p}\Big) }.
	\end{equation}
	
	After transmitting the Msg1, the device starts to monitor PDCCH containing DCI continuously. The device monitors the PDCCH in each slot for $ \mathrm{T_{d}} $, and then stays in sleep state for the rest of the time $ \mathrm{T_{s}-T_{d}}  $. We assume the Msg2 is scheduled following PDCCH in the same slot, and the Msg2 from the BS arrives at the half duration of the RAR window in the successful case (i.e., $ \frac{\mathrm{N_{RAR}}}{2}*\mathrm{T_{s}} $). After receiving the Msg2, the device needs to wait for $ \mathrm{T_{K2}+T_{\triangle}}  $ before transmitting the Msg3 as Fig.~\ref{fig:Timing relationship of 4-step and NR EDT RA schemes}. Therefore, the energy consumption of step 2 in 4-step RA scheme when preamble detection succeeds can be written as
	\begin{equation}
	\begin{aligned}
	&E_{\mathrm{Msg2s}} = \Big( {\mathrm{N_{RAR}}}/{2}-1 \Big)\times \\&\bigg( \mathrm{P_{r}T_{d}}+\mathrm{P_{s}\Big( T_{s}-T_{d}\Big) }\bigg) +\mathrm{P_{r}T_{s}}  +\mathrm{P_{s}\Big( T_{K2}+T_{\triangle}\Big) }.
	\end{aligned}
	\end{equation}
	
	Then, the energy consumption of transmitting Msg3 at step 3 in 4-step RA scheme can be written as
	\begin{equation}
	E_{\mathrm{Msg3}} = \mathrm{P_{t}T_{s}}.
	\end{equation}
	
	After transmitting the Msg3, the device starts the CRT and monitors the PDCCH for Msg4 DCI. We assume the Msg4 from the BS arrives at the half duration of the CRT in the successful case (i.e., $ \frac{\mathrm{N_{CRT}}}{2}*\mathrm{T_{s}} $). The device that succeeds in the contention receives the Msg4, and then transmits the ACK after $ \mathrm{T_{PUCCH}} $ as Fig.~\ref{fig:Timing relationship of 4-step and NR EDT RA schemes}. Therefore, the energy consumption of step 4 in the 4-step RA scheme when PUSCH decoding succeeds can be written as
	
	\begin{equation}
	\begin{aligned}
	&E_{\mathrm{Msg4s}} =\Big( {\mathrm{N_{CRT}}}/{2}-1\Big)\times\\&\bigg( \mathrm{P_{r}T_{d}}+\mathrm{P_{s}\Big( T_{s}-T_{d}\Big)  }\bigg)+\mathrm{P_{r}T_{s}}+\mathrm{T_{{PUCCH}}}\mathrm{P_{s}}+\mathrm{P_{t}T_{s}}.
	\end{aligned}
	\end{equation}

	The HARQ is applied to data transmission after successful 4-step RA procedure (marked with dash line in Fig.~\ref{fig:Timing relationship of 4-step and NR EDT RA schemes}) to guarantee the reliability. The device first receives the data DCI, which lasts for $ \mathrm{T_{d}} $. Then, the device transmits data after $ \mathrm{T_{K2}} $, and starts to monitor the RRC DCI. We assume the RRC DCI arrives after $\mathrm{N_{DCI}}$ time slots, which contains the new data indicator (NDI) to indicate the data transmission failure or success.\footnote{The $\mathrm{N_{DCI}}$ time slots are the gap between the uplink data transmission and downlink RRC DCI, which consists of decoding latency and L2L1 processing delay \cite{Patriciello2019}.}  If the data transmission succeeds, the device transmits ACK after $ \mathrm{T_{PUCCH}} $. Therefore, the energy consumption when the data HARQ completes in $ k $th transmissions can be written as
	\begin{equation}
	\begin{aligned}
	&E_{\mathrm{data}}^k = k\bigg[\mathrm{P_{r}T_{d}}+\mathrm{P_{s}\Big( T_{s}-T_{d}\Big)  }  +\mathrm{T_{K2}P_{s}}+\\&\mathrm{P_{t}T_{s}}+\mathrm{N_{DCI}}\bigg( \mathrm{P_{r}T_{d}}+\mathrm{P_{s}\Big( T_{s}-T_{d}\Big)  }\bigg)+\mathrm{P_{r}T_{s}} \bigg]   +\\&\mathrm{T_{{PUCCH}}}\mathrm{P_{s}}+\mathrm{P_{t}T_{s}}.
	\end{aligned}
	\end{equation}
	
	The average data transmission energy consumption after successful 4-step RA procedure is derived as
	\begin{equation}
	E_{\mathrm{data}}^{\mathrm{HARQ}}=\sum_{k=1}^{\mathrm{K-1}}(1-B) B^{k-1}E_{\mathrm{data}}^{k}+ B^{\mathrm{K}-1}E_{\mathrm{data}}^{\mathrm{K}},
	\end{equation}
	where $ \mathrm{K} $ is the maximum data retransmission times.

	\subsubsection{Preamble Detection Failure}
	The energy consumption of step 1 is the same as that of the successful 4-step RA procedure. If preamble detection fails at the step 1 of 4-step RA scheme, the device will not receive the Msg2, which means the device monitors PDCCH until the end of RAR window and waits for the next RACH opportunity. Hence, the energy consumption of step 2 in 4-step RA scheme when preamble detection fails can be written as
	\begin{equation}
	E_{\mathrm{Msg2f}} = \mathrm{N_{RAR}}\bigg( \mathrm{P_{r}T_{d}}+\mathrm{P_{s}\Big( T_{s}-T_{d}\Big)  }\bigg).
	\end{equation}
	
	\subsubsection{PUSCH Decoding Failure}
	 The energy consumption of step 1, 2, and 3 are the same as that of the successful 4-step RA procedure. If PUSCH decoding fails at the step 3 of the 4-step RA scheme, the device will not receive Msg4, which means the device monitors PDCCH until the end of CRT and waits for the next RACH opportunity. Hence, the energy consumption of step 4 in 4-step RA scheme when PUSCH decoding fails can be written as
	\begin{equation}
	E_{\mathrm{Msg4f}}= \mathrm{N_{CRT}}\bigg( \mathrm{P_{r}T_{d}}+\mathrm{P_{s}\Big( T_{s}-T_{d}\Big)  }\bigg).
	\end{equation}

	Finally, the average energy consumption of a typical device in $ m $th RACH attempt of 4-step RA scheme is derived as
	\begin{equation}
	\begin{aligned}
	\begin{split}
	&E^{m}_{\mathrm{4step}}=\mathcal{T}^{m}_{\mathrm{4step}}\biggl\{\underbracea{\left(  1-\sum_{n=0}^{\infty}\mathbb{P}\left[N=n\right] {\cal P}_{{\rm{pre}|\mathcal{A}}}^m\right) \Big(E_{\mathrm{p}}+  }_{\mathrm{I}}\\&\underbraced{E_{\mathrm{Msg2f}}\Big)}_{\mathrm{I}}+ \underbracea{\sum_{n=0}^{\infty}\mathbb{P}\left[N=n \right] {\cal P}_{{\rm{pre}|\mathcal{A}}}^m\left(1-{\cal P}_{{\rm{pus}|\mathcal{B}}}^{m}\right)\Big(E_{\mathrm{p}}+ }_{\mathrm{II}}\\&\underbraced{E_{\mathrm{Msg2s}}+ E_{\mathrm{Msg3}}  +E_{\mathrm{Msg4f}}\Big) }_{\mathrm{II}} +\underbracea{ \sum_{n=0}^{\infty}\mathbb{P}\left[N=n \right] {\cal P}_{{\rm{pre}|\mathcal{A}}}^m{\cal P}_{{\rm{pus}|\mathcal{B}}}^{m}}_{\mathrm{III}}\\&\underbraced{\times\Big( E_{\mathrm{p}}+ E_{\mathrm{Msg2s}}+E_{\mathrm{Msg3}}+  E_{\mathrm{Msg4s}}+ E_{\mathrm{data}}^{\mathrm{HARQ}}\Big)}_{\mathrm{III}}\biggr\},
	\end{split}
	\label{eq:overall_e_4step}
	\end{aligned}
	\end{equation}
	where term I, II, and III are the energy consumption of 4-step RA scheme under preamble detection failure, PUSCH decoding failure, and successful 4-step RA procedure, $ \mathcal{T}^{m}_{\mathrm{4step}} $ is the non-empty probability of a typical device in $ m $th RACH attempt, $ \mathcal{P}^{m}_{\mathrm{pre|A}} $ is the preamble detection success probability given in \textbf{Lemma 1}, $ \mathcal{P}^{m}_{\mathrm{pus|B}} $ is the PUSCH decoding success probability given in \textbf{Lemma 2}.

	\subsection{4-step SDT Random Access}
	As shown in Fig.~\ref{fig:Timing relationship of 4-step and NR EDT RA schemes}, 4-step SDT RA scheme grants a larger PUSCH resource for Msg3 transmission along with the data (marked with dash line). Similar as 4-step RA scheme, the energy consumption of 4-step SDT also depends on the success or failure of preamble detection and PUSCH decoding. The difference is that the data is transmitted along with Msg3 without HARQ, whose energy consumption is $  E_{\mathrm{data}}= \mathrm{P_{t}T_{s}}$. Hence, the average energy consumption of a typical device in $ m $th RACH attempt in 4-step SDT RA scheme is derived as
	\begin{equation}
	\begin{split}
	&E^{m}_{\mathrm{4-stepSDT}}=\mathcal{T}^{m}_{\mathrm{4-stepSDT}}\biggl\{\underbracea{\left(  1-\sum_{n=0}^{\infty}\mathbb{P}\left[N=n\right] {\cal P}_{{\rm{pre}|\mathcal{A}}}^m\right)\times }_{\mathrm{I}}\\&\underbraced{\Big(E_{\mathrm{p}}+  E_{\mathrm{Msg2f}}\Big)}_{\mathrm{I}}+ \underbracea{\sum_{n=0}^{\infty}\mathbb{P}\left[N=n \right] {\cal P}_{{\rm{pre}|\mathcal{A}}}^m\left(1-{\cal P}_{{\rm{pus}|\mathcal{B}}}^{m}\right)\times}_{\mathrm{II}}\\&\underbraced{\Big(  E_{\mathrm{p}}+ E_{\mathrm{Msg2s}}+ E_{\mathrm{Msg3}}+E_{\mathrm{data}}+E_{\mathrm{Msg4f}}\Big) }_{\mathrm{II}} +\underbracea{ \sum_{n=0}^{\infty}\mathbb{P}\left[N=n \right]}_{\mathrm{III}}\\&\underbraced{\times{\cal P}_{{\rm{pre}|\mathcal{A}}}^m{\cal P}_{{\rm{pus}|\mathcal{B}}}^{m}\Big( E_{\mathrm{p}}+ E_{\mathrm{Msg2s}}+ E_{\mathrm{Msg3}}+ E_{\mathrm{data}}+  E_{\mathrm{Msg4s}}\Big)}_{\mathrm{III}}\biggr\},
	\end{split}
	\label{eq:overall_e_EDT}
	\end{equation}
	where term I, II, and III are the average energy consumption of 4-step SDT under preamble detection failure, PUSCH decoding failure, and successful RA procedure. Since data is transmitted along with Msg3, both term II and term III include data energy consumption $  E_{\mathrm{data}}$.

	\subsection{2-step Random Access}
 
	The average energy consumption of 2-step RA scheme depends on the success or failure of preamble detection and PUSCH decoding in MsgA, and PUSCH decoding after fallback. Therefore, we derive the energy consumption of each message in 2-step RA scheme with successful 2-step RA procedure, preamble detection failure, PUSCH decoding failure, respectively.
		\begin{figure*}[htbp!]
		\centering
		\includegraphics[scale=0.8]{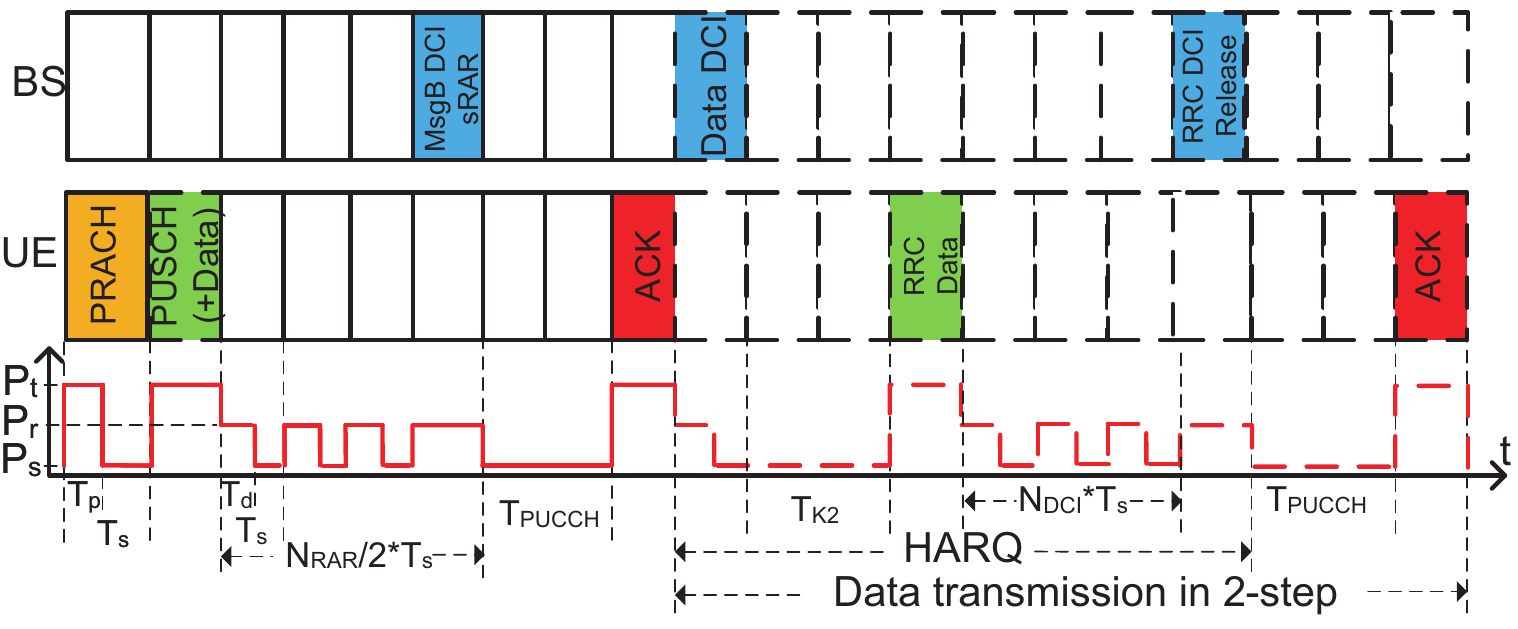}
		\caption{Timing relationship of successful 2-step and 2-step SDT RA procedure.}
		\label{fig:Timing relationship of 2-step A and 2-step B RA schemes}
	\end{figure*}
	\subsubsection{Successful 2-step RA Procedure}
	The timing relationship of successful 2-step RA procedure is illustrated in Fig.~\ref{fig:Timing relationship of 2-step A and 2-step B RA schemes}. The preamble transmission lasts for $ \mathrm{T_p} $ during $ \mathrm{T_s} $, and PUSCH transmission lasts for $ \mathrm{T_s} $. Hence, we derive the energy consumption of step 1 in 2-step RA scheme as
	\begin{equation}
	E_{\mathrm{MsgA}} =\mathrm{P_{t}T_{p}}+\mathrm{P_{s}\Big( T_{s}-T_{p}\Big) } +\mathrm{P_{t}T_{s}}.
	\end{equation}

	After transmitting the MsgA, the device starts to monitor PDCCH containing DCI continuously for $ \mathrm{T_{d}} $ in each slot. We assume the MsgB from the BS arrives at the half duration of the RAR window in the successful case (i.e., $ \frac{\mathrm{N_{RAR}}}{2}*\mathrm{T_{s}} $). The device that succeeds in the contention receives the MsgB, and then transmits the ACK after $ \mathrm{T_{PUCCH}} $ as Fig.~\ref{fig:Timing relationship of 2-step A and 2-step B RA schemes}. Therefore, the energy consumption of step 2 in 2-step RA scheme with a successful MsgB is derived as
	\begin{equation}
	\begin{aligned}
	&E_{\mathrm{MsgBs}} =\Big( {\mathrm{N_{RAR}}}/{2}-1\Big)\times\\&\bigg( \mathrm{P_{r}T_{d}}+\mathrm{P_{s}\Big( T_{s}-T_{d}\Big)  }\bigg)+\mathrm{P_{r}T_{s}}+\mathrm{T_{{PUCCH}}}\mathrm{P_{s}}+\mathrm{P_{t}T_{s}}.
	\end{aligned}
	\end{equation}.
	
	\subsubsection{Preamble Detection Failure}
	The energy consumption of step 1 is the same as that of the successful 2-step RA procedure. If preamble detection fails at the step 1 of the 2-step RA scheme, the device will not receive MsgB, which means the device monitors PDCCH until the end of RAR window and waits for the next RACH opportunity. Hence, the energy consumption of step 2 in 2-step RA scheme when preamble detection fails can be written as
	\begin{equation}
	E_{\mathrm{MsgBf}} = \mathrm{N_{RAR}}\bigg( \mathrm{P_{r}T_{dci}}+\mathrm{P_{s}\Big( T_{s}-T_{d}\Big)  }\bigg).
	\end{equation}
		
	\subsubsection{PUSCH Decoding Failure}
	The energy consumption of step 1 is the same as that of the successful 2-step RA procedure. If the BS successfully detects the preamble in MsgA but fails to decode the PUSCH at the step 1 of the 2-step RA scheme, the device will receive a fallback MsgB with grant for Msg3 transmission following 4-step RA procedure as Fig.~\ref{fig:Timing relationship of 4-step and NR EDT RA schemes}. We assume the MsgB from the BS arrives at the half duration of the RAR window, and the device transmits Msg3 after $ \mathrm{T_{K2}+T_{\triangle}} $. Hence, we derive the energy consumption of step 2 in 2-step RA scheme when fallback mechanism is satisfied as
	\begin{equation}
	\begin{aligned}
	&E_{\mathrm{MsgBfb}} = \Big( {\mathrm{N_{RAR}}}/{2}-1\Big)\times\\& \bigg( \mathrm{P_{r}T_{d}}+\mathrm{P_{s}\Big( T_{s}-T_{d}\Big)  }\bigg) +\mathrm{P_{r}T_{s}} +\mathrm{P_{s}\Big( T_{K2}+T_{\triangle}\Big) }.
	\end{aligned}
	\end{equation}
	
	Since the data can be transmitted after successful 2-step or successful fallback 4-step RA procedure with HARQ, the average energy consumption of a typical device in $ m $th RACH attempt in 2-step RA scheme is derived as
	\begin{equation}
	\begin{split}
	&E^{m}_{\mathrm{2step}}=\mathcal{T}^{m}_{\mathrm{2step}}\biggl\{\underbracea{\sum_{n=0}^{\infty}\mathbb{P}\left[N=n \right] {\cal P}_{{\rm{pre}|\mathcal{A}}}^m{\cal P}_{{\rm{pus}|\mathcal{B}}}^{m} \Big( E_{\mathrm{MsgA}}+ }_{\mathrm{I}}\\&\underbraced{E_{\mathrm{MsgBs}}+E_{\mathrm{data}}^{\mathrm{HARQ}}\Big)}_{\mathrm{I}} +\underbracea{\bigg( 1-\sum_{n=0}^{\infty}\mathbb{P}\left[N=n\right] {\cal P}_{{\rm{pre}|\mathcal{A}}}^m\Big( 1-n\times}_{\mathrm{II}}\\&\underbraced{{\cal P}_{{\rm{pus}|\mathcal{B}}}^m\Big)\bigg)\left( E_{\mathrm{MsgA}}+ E_{\mathrm{MsgBf}}\right)}_{\mathrm{II}} +\underbracea{\sum_{n=0}^{\infty}\mathbb{P}\left[N=n\right] {\cal P}_{{\rm{pre}|\mathcal{A}}}^m\Big( 1-}_{\mathrm{III}}\\&\underbracec{\left(n+1 \right) {\cal P}_{{\rm{pus}|\mathcal{B}}}^m\Big){\cal P}_{{\rm{pus}|\mathcal{B}}}^m\Big( E_{\mathrm{MsgA}}+ E_{\mathrm{MsgBfb}}+E_{\mathrm{Msg3}}+ }_{\mathrm{III}}\\&\underbraced{E_{\mathrm{Msg4s}}+ E_{\mathrm{data}}^{\mathrm{HARQ}}\Big)}_{\mathrm{III}}+\underbracea{\sum_{n=0}^{\infty}\mathbb{P}\left[N=n\right] {\cal P}_{{\rm{pre}|\mathcal{A}}}^m\Big( 1-\left(n+1 \right)\times}_{\mathrm{IV}}\\&\underbracec{{\cal P}_{{\rm{pus}|\mathcal{B}}}^m  \Big)\left(1- {\cal P}_{{\rm{pus}|\mathcal{B}}}^m\right) \Big( E_{\mathrm{MsgA}}+ E_{\mathrm{MsgBfb}}+E_{\mathrm{Msg3}} + }_{\mathrm{IV}}\\&\underbraced{E_{\mathrm{Msg4f}}\Big)}_{\mathrm{IV}} \biggr\},
	\end{split}
	\label{eq:overall_e_2stepB}
	\end{equation}
	where term I, II, III, IV correspond to the successful 2-step RA procedure, preamble detection failure, successful fallback 4-step RA procedure, and PUSCH decoding failure in fallback 4-step RA procedure, respectively.

		\subsection{2-step SDT Random Access}

	Similar as the 2-step RA scheme, the energy consumption of 2-step SDT RA scheme also depends on the success or failure of preamble detection and PUSCH decoding. The difference is that the data is transmitted along with MsgA without HARQ, whose energy consumption is $  E_{\mathrm{data}}= \mathrm{P_{t}T_{s}}$. Hence, the average energy consumption of a typical device in $ m $th RACH attempt in 2-step SDT RA scheme is derived as 
		
	\begin{equation}
	\begin{split}
	&E^{m}_{\mathrm{2stepSDT}}=\mathcal{T}^{m}_{\mathrm{2stepSDT}}\biggl\{\underbracea{\sum_{n=0}^{\infty}\mathbb{P}\left[N=n \right] {\cal P}_{{\rm{pre}|\mathcal{A}}}^m{\cal P}_{{\rm{pus}|\mathcal{B}}}^{m} \Big(  E_{\mathrm{MsgA}}}_{\mathrm{I}}\\&\underbraced{+E_{\mathrm{data}}+ E_{\mathrm{MsgBs}}\Big)}_{\mathrm{I}} +\underbrace{\bigg( 1-\sum_{n=0}^{\infty}\mathbb{P}\left[N=n\right] {\cal P}_{{\rm{pre}|\mathcal{A}}}^m\Big( 1-n\times }_{\mathrm{II}}\\&\underbraced{{\cal P}_{{\rm{pus}|\mathcal{B}}}^m\Big)\bigg)\Big( E_{\mathrm{MsgA}}+E_{\mathrm{data}}+ E_{\mathrm{MsgBf}}\Big)}_{\mathrm{II}} +\underbracea{\sum_{n=0}^{\infty}\mathbb{P}\left[N=n\right]\times }_{\mathrm{III}}\\&\underbracec{{\cal P}_{{\rm{pre}|\mathcal{A}}}^m\left( 1-\left(n+1 \right) {\cal P}_{{\rm{pus}|\mathcal{B}}}^m\right){\cal P}_{{\rm{pus}|\mathcal{B}}}^m\Big( E_{\mathrm{MsgA}}+E_{\mathrm{data}}+ }_{\mathrm{III}}\\&\underbraced{E_{\mathrm{MsgBfb}}+E_{\mathrm{Msg3}} +E_{\mathrm{Msg4s}}+ E_{\mathrm{data}}^{\mathrm{HARQ}}\Big)}_{\mathrm{III}}+\underbracea{
		\sum_{n=0}^{\infty}\mathbb{P}\left[N=n\right]\times  }_{\mathrm{IV}}\\&\underbracec{{\cal P}_{{\rm{pre}|\mathcal{A}}}^m\left( 1-\left(n+1 \right) {\cal P}_{{\rm{pus}|\mathcal{B}}}^m\right)\left(1- {\cal P}_{{\rm{pus}|\mathcal{B}}}^m\right)\Big( E_{\mathrm{MsgA}}+E_{\mathrm{data}}}_{\mathrm{IV}}\\&\underbraced{+ E_{\mathrm{MsgBfb}}+E_{\mathrm{Msg3}} +E_{\mathrm{Msg4f}}\Big)}_{\mathrm{IV}} \biggr\},
	\label{eq:overall_e_2stepA}
	\end{split}
	\end{equation}
	where term I, II, III, IV correspond to the successful 2-step SDT RA procedure, preamble detection failure, successful fallback 4-step RA procedure, and PUSCH decoding failure in fallback 4-step RA procedure.

	\subsection{Average Energy Consumption for Each Packet Transmission}
	We have derived the average energy consumption of each scheme in $ m $th RACH attempt, however, the average energy consumption for a successful packet transmission in each scheme is a more important metric. Therefore, We derive the average energy consumption of a successful packet transmission as
	\begin{equation}
	\begin{aligned}
	\mathbb{E}\left[E^{m}\right]={\sum_{t=1}^{m}E^{m}}/{\sum_{t=1}^{m}\mathcal{T}^{m}{\cal P}^m},
	\end{aligned}
	\label{eq:average energy consumption}
	\end{equation}
	where $ E^{m} $ is the average energy consumption in $ m $th RACH attempt of each scheme, which has been derived for each scheme in \eqref{eq:overall_e_4step}, \eqref{eq:overall_e_EDT}, \eqref{eq:overall_e_2stepB}, \eqref{eq:overall_e_2stepA}.

	\section{Simulation and Discussion}
	In this section, we validate our analysis above via independent system level simulations based on Monte Carlo method and the queue evolution over consecutive time slots. As mentioned in Section II, the devices are deployed via independent PPPs in a 0.1 $ \mathrm{km}^2 $ circle cell. We assume that the frequency band only impacts the slot duration. Each device employs the channel inversion power control, and the buffer at each device is simulated to capture the new packets arrival and accumulation process evolved along the time. The preamble detection, PUSCH decoding, and data transmission processes are jointly simulated to verify the packet transmission success probability, average energy consumption for a packet transmission, and throughput derived in section III and IV. Remind that the advanced receiver model enables the BS to decode the PUSCH signal with the highest SINR above the threshold, even in the presence of simultaneous transmission. We compare the results with the basic receiver in \cite{Jiang2018a}, where all transmissions are failed when the collision happens. In all figures of this section, we use “Ana.” and “Sim.” to abbreviate “Analytical” and “Simulation”, respectively. Unless otherwise stated, we set the same new packets arrival rate for each time slot as $ \mu^{m}_{\mathrm{New}} =0.1$ packets/time-slot, $ \rho=-90 $ dBm, $ \sigma_{\mathrm{n}}^{2} = -100.4 $ dBm, $ \gamma_{\mathrm{ th}}=-10 $ dB, $ \alpha=4 $, $ \mathrm{T_{p}}=(133 + 9.4)$ us , $ \mathrm{T_{slot}} =0.5$ ms, $  \mathrm{T_{dci}} =107.14$ us , $ \mathrm{N_{RAR}}=40 $ , $ \mathrm{N_{CRT}}=48 $. $ \mathrm{T_{K2}}=0.5 $ ms, $ \mathrm{T_{\triangle}} =1.5$ ms,  $ \mathrm{T_{PUCCH}}=964.29 $ us,  $ \mathrm{P_{s}} $=15 uW,  $ \mathrm{P_{r}}=80$ mW,  $ \mathrm{P_{t}}=500 $ mW, $ N_{\mathrm{ZC}}=839 $, $ \lambda _{{\rm{th}}} =-51.5$ dBm, $ \mathrm{K}=1 $, $ \mathrm{T_{{RACH}}} = 31.5$ms.
	
	\begin{figure}[htbp!]
				\centering
				\includegraphics[scale=0.35]{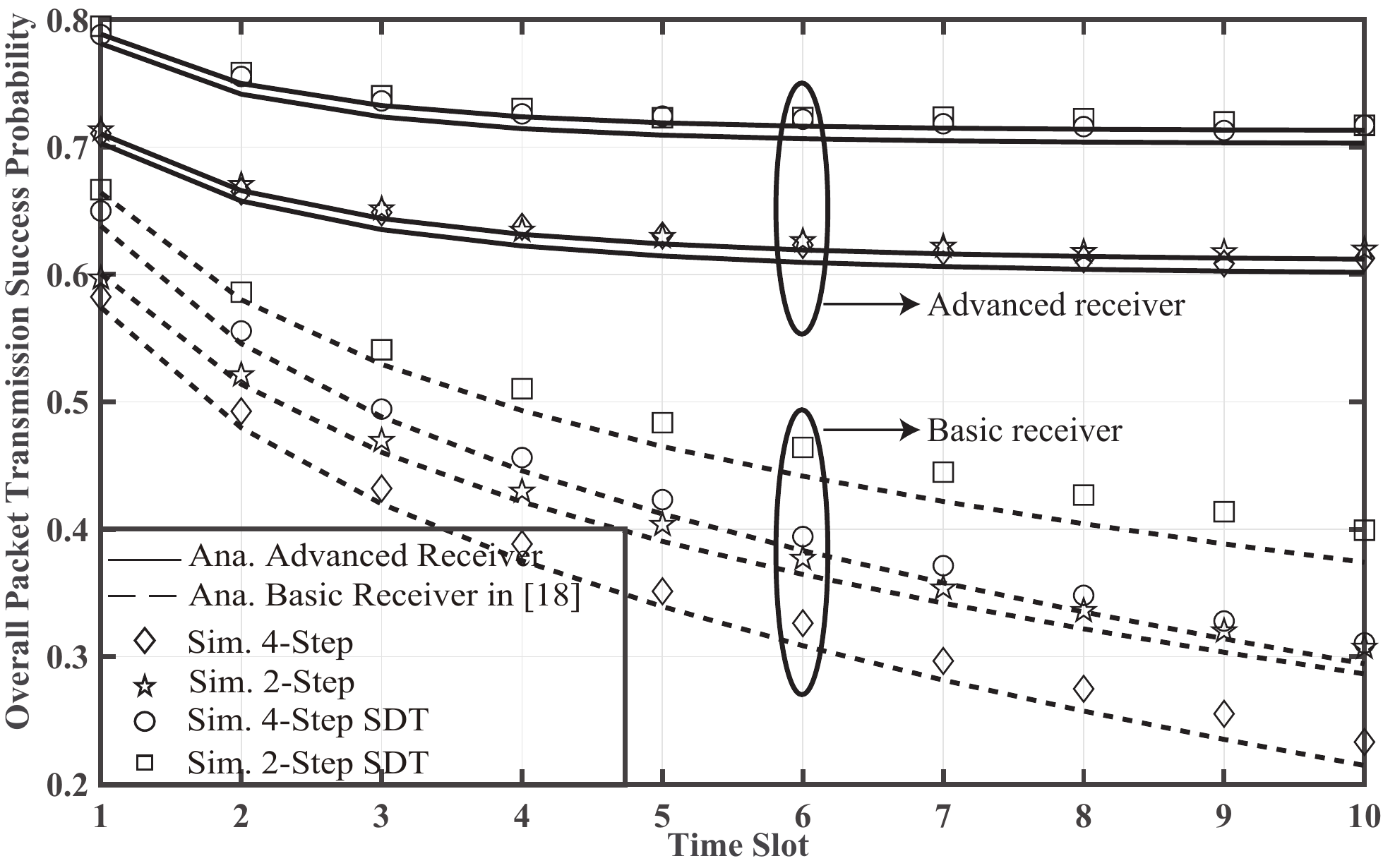}
				\caption{Packet transmission success probability in each time slot. We set the device intensity $ \lambda_{\mathrm{Dp}}= $5 devices/preamble}
				\label{fig: success probability in each time slot}
	\end{figure}
	
		\begin{figure}[htbp!]
			\centering
			\includegraphics[scale=0.35]{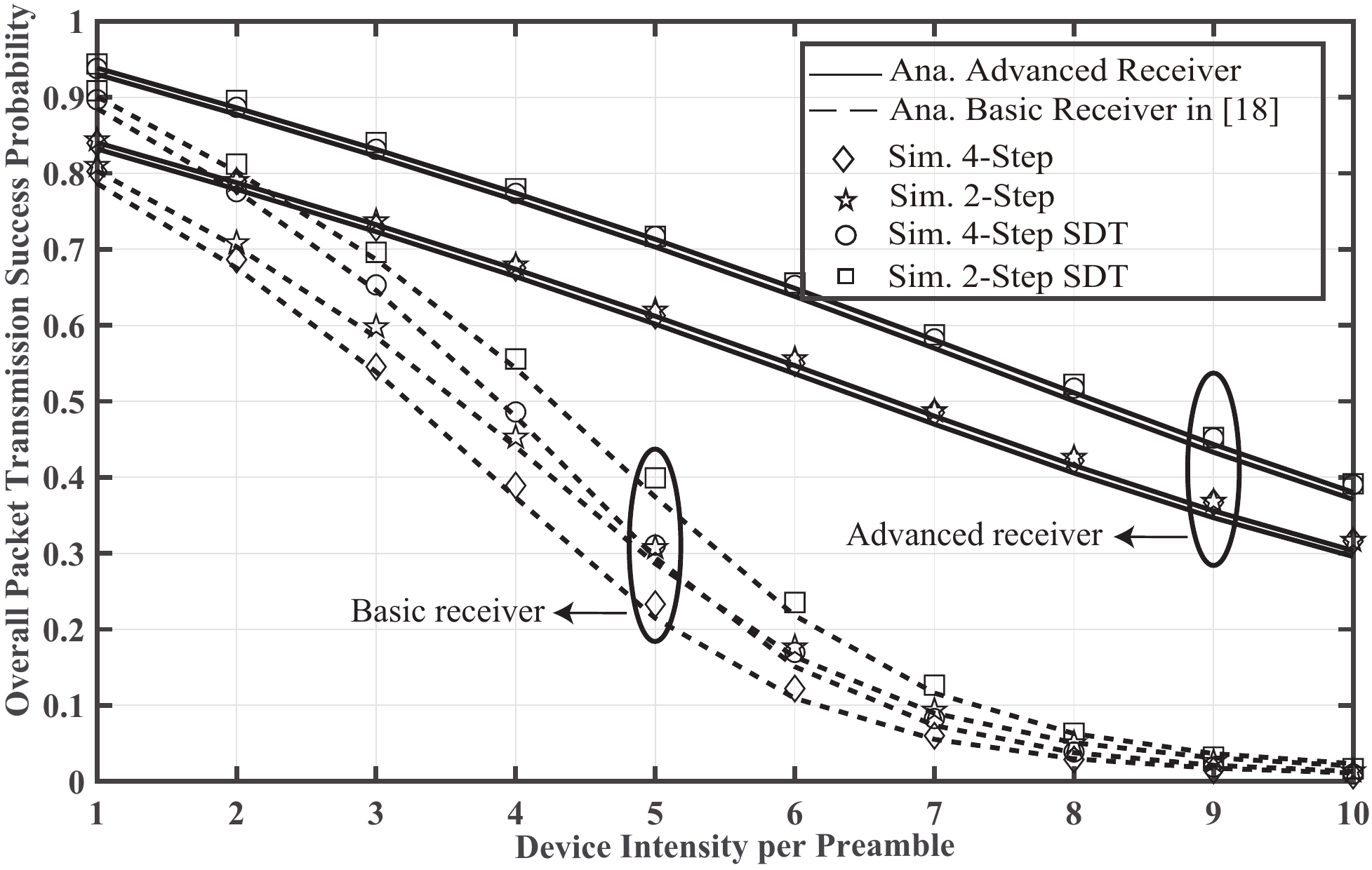}
			\caption{Packet transmission success probability in $ 10 $th time slot with different device intensity}
			\label{fig: Overall success probability}
				\end{figure}

	Fig.~\ref{fig: success probability in each time slot} plots the overall packet transmission success probability with four RA schemes versus time slot with $ \lambda_{\mathrm{Dp}}=5 $. The analytical curves of 4-step, 4-step SDT, 2-step, and 2-step SDT RA schemes are plotted using \eqref{eq:four step overall success probability}, \eqref{eq:EDT overall success probability}, \eqref{eq:2stepA_overall} and \eqref{eq:2stepB_overall}. The close match between the analytical curves and simulation points validates the accuracy of our developed spatio-temporal mathematical framework. In Fig.~\ref{fig: success probability in each time slot}, we see that the overall packet transmission success probability of each scheme enters into the stable region in the advanced receiver. However, the overall packet transmission success probability of each scheme keeps decreasing as time evolves in the basic receiver. This is because the BS can not decode the PUSCH if multiple colliding devices' SINR is higher than the threshold in the basic receiver. Therefore, the new arrival packets can not be transmitted to the BS in time, and leads to the traffic congestion.

	We also observe that the overall packet transmission success probabilities of 4/2-step SDT RA schemes are higher than that of 4/2-step RA schemes. This can be explained by the reason that the data is transmitted after RACH procedures in 4/2-step RA schemes, which decreases the overall packet transmission success probability. We also notice that the success probabilities of 2-step SDT RA scheme is higher than that of 4-step SDT RA scheme, and 2-step RA scheme is higher than that of 4-step RA scheme. This is because 2-step and 2-step SDT RA schemes integrate the fallback mechanism, which enables the packet to be successfully transmitted after fallback to 4-step RA scheme. This advantage is more obvious in the basic receiver due to higher fallback probability.

	Fig.~\ref{fig: Overall success probability} plots the overall packet transmission success probability with four RA schemes versus device intensity in $ 10 $th time slot.  In Fig.~\ref{fig: Overall success probability}, we observe that the overall packet transmission success probability of all schemes decreases with the increasing device intensity, due to the increasing aggregate interference from more devices transmitting signals simultaneously. We also observe that the decreasing rate is almost linear in the advanced receiver, and more sharply in the basic receiver due to more serious traffic congestion. 
	
	\begin{figure}[htbp!]
				\centering
		\includegraphics[scale=0.35]{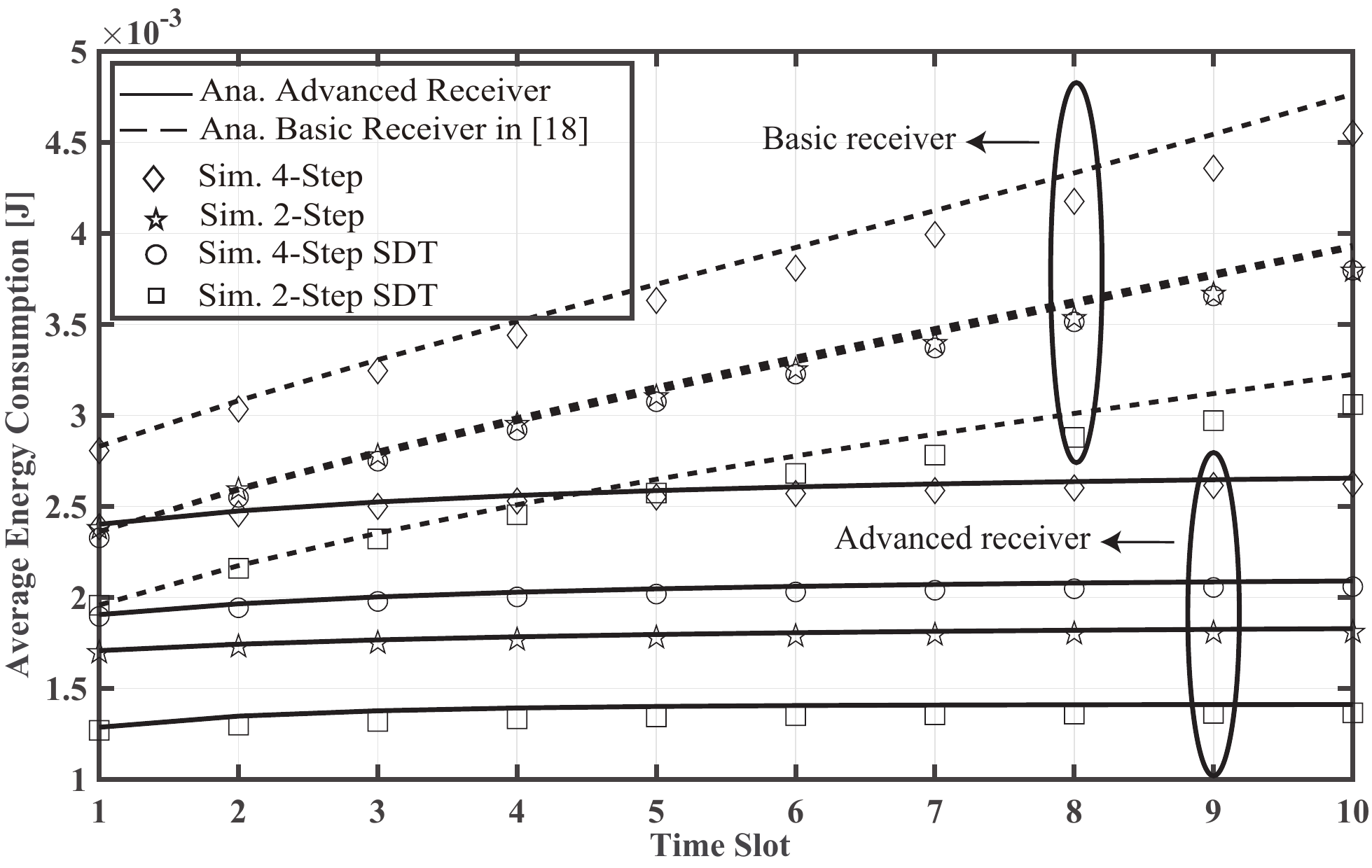}
\caption{Average energy consumption for packet transmission in each time slot. We set the device intensity $ \lambda_{\mathrm{Dp}}= $5 devices/preamble}
\label{fig:Energy consumption in each time slot}
	\end{figure}
	
		\begin{figure}[htbp!]
							\centering
			\includegraphics[scale=0.35]{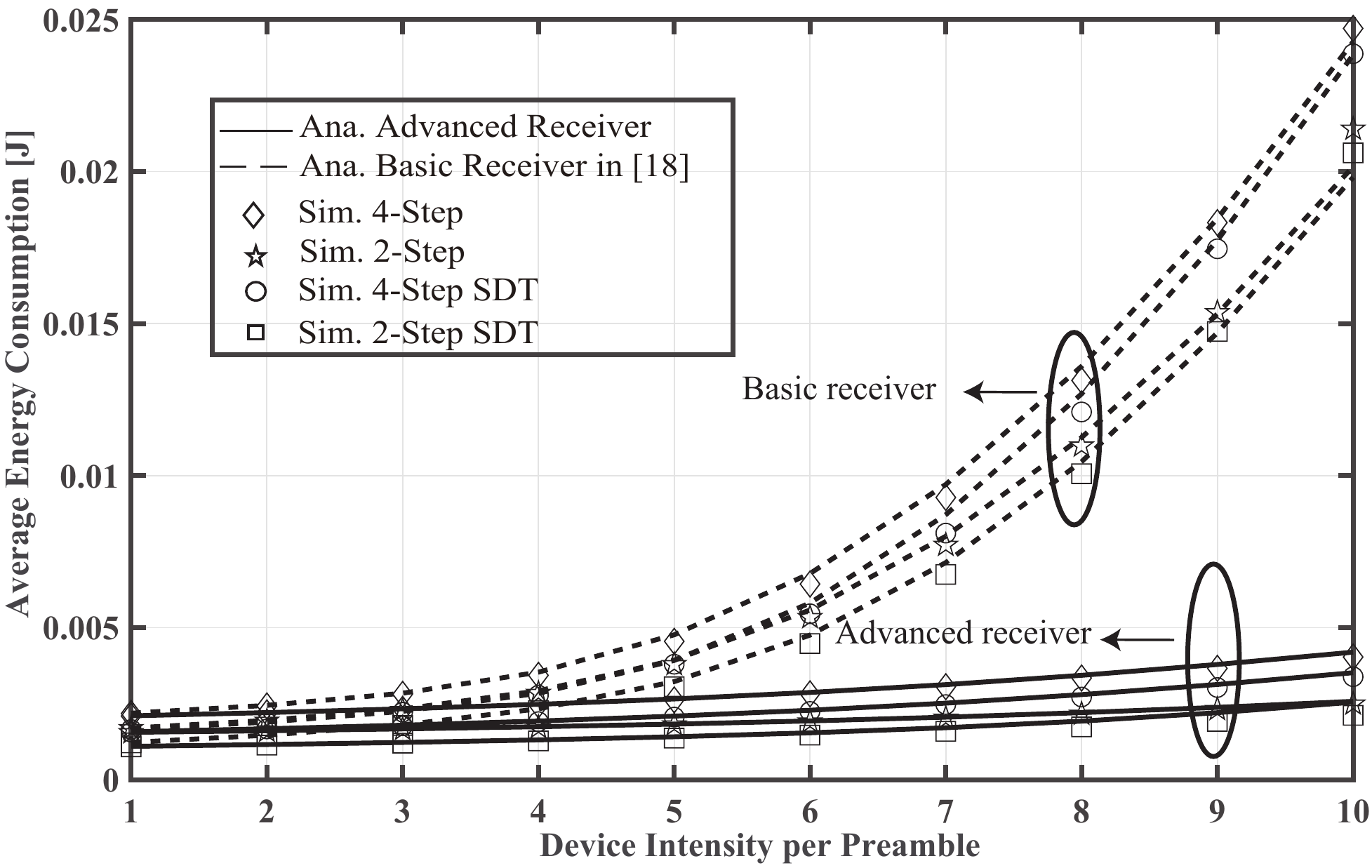}
			\caption{Average energy consumption for packet transmission in $ 10 $th time slot with different device intensity}
			\label{fig:Energy consumption}
				\end{figure}

	Fig.~\ref{fig:Energy consumption in each time slot} plots the average energy consumption for a successful packet transmission with four RA schemes versus time slot. The analytical curves of 4-step, 4-step SDT, 2-step and 2-step SDT RA schemes are plotted using \eqref{eq:overall_e_4step}, \eqref{eq:overall_e_EDT}, \eqref{eq:overall_e_2stepB}, \eqref{eq:overall_e_2stepA} and \eqref{eq:average energy consumption}. In Fig.~\ref{fig:Energy consumption in each time slot}, we observe that the average energy consumption for a successful packet transmission reaches a stable state in the advanced receiver, however, the average energy consumption keeps increasing with time-evolving in the basic receiver due to low overall packet transmission probability.

	Fig.~\ref{fig:Energy consumption} plots the average energy consumption for a successful packet transmission with four RA schemes versus device intensity. In Fig.~\ref{fig:Energy consumption}, we see that the average energy consumption follows 4-step $ > $ 4-step SDT $ > $ 2-step $ > $ 2-step SDT RA schemes. However, we can also notice that the gap between 2-step and 2-step SDT RA schemes is decreasing with increasing device intensity. This is because data is transmitted during RA procedure in 2-step SDT RA scheme, which leads to energy waste in high device intensity scenario.  We can also observe that the average energy consumption gap between GF (i.e., 2-step and 2-step SDT) and GB (i.e., 4-step and 4-step SDT) RA schemes increases with the increase of device intensity. The reason is that re-attempt using 4-step and 4-step SDT is more expensive in terms of energy consumption than a re-attempt using 2-step and 2-step SDT RA schemes due to the longer RA procedures.
	
		\begin{figure}[htbp!]
				\centering
		\includegraphics[scale=0.35]{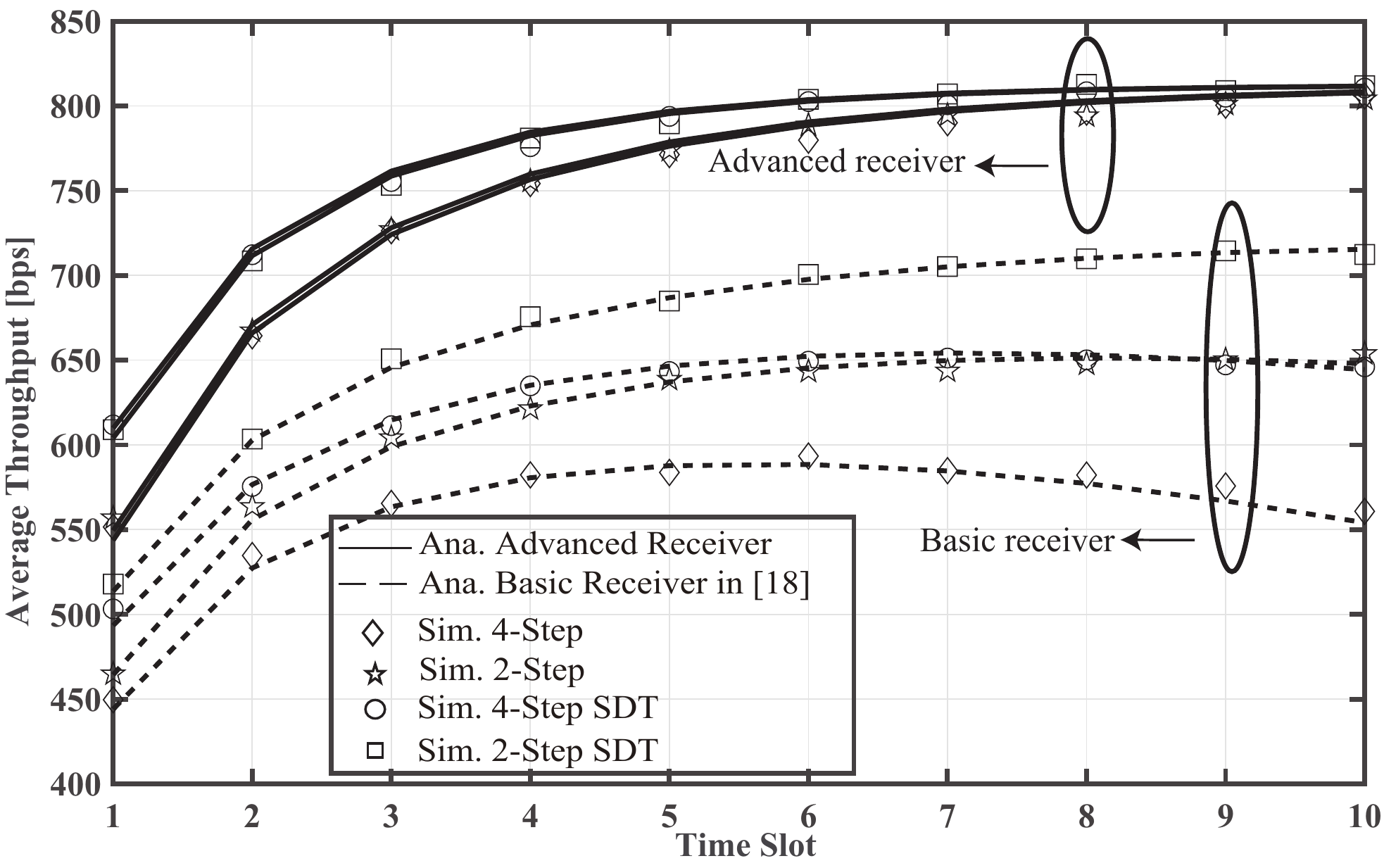}
\caption{Average throughput in each time slot. We set the device intensity $ \lambda_{\mathrm{Dp}}= $5 devices/preamble}
\label{fig:Average throughput in each time slot}
	\end{figure}
	
			\begin{figure}[htbp!]
								\centering
				\includegraphics[scale=0.35]{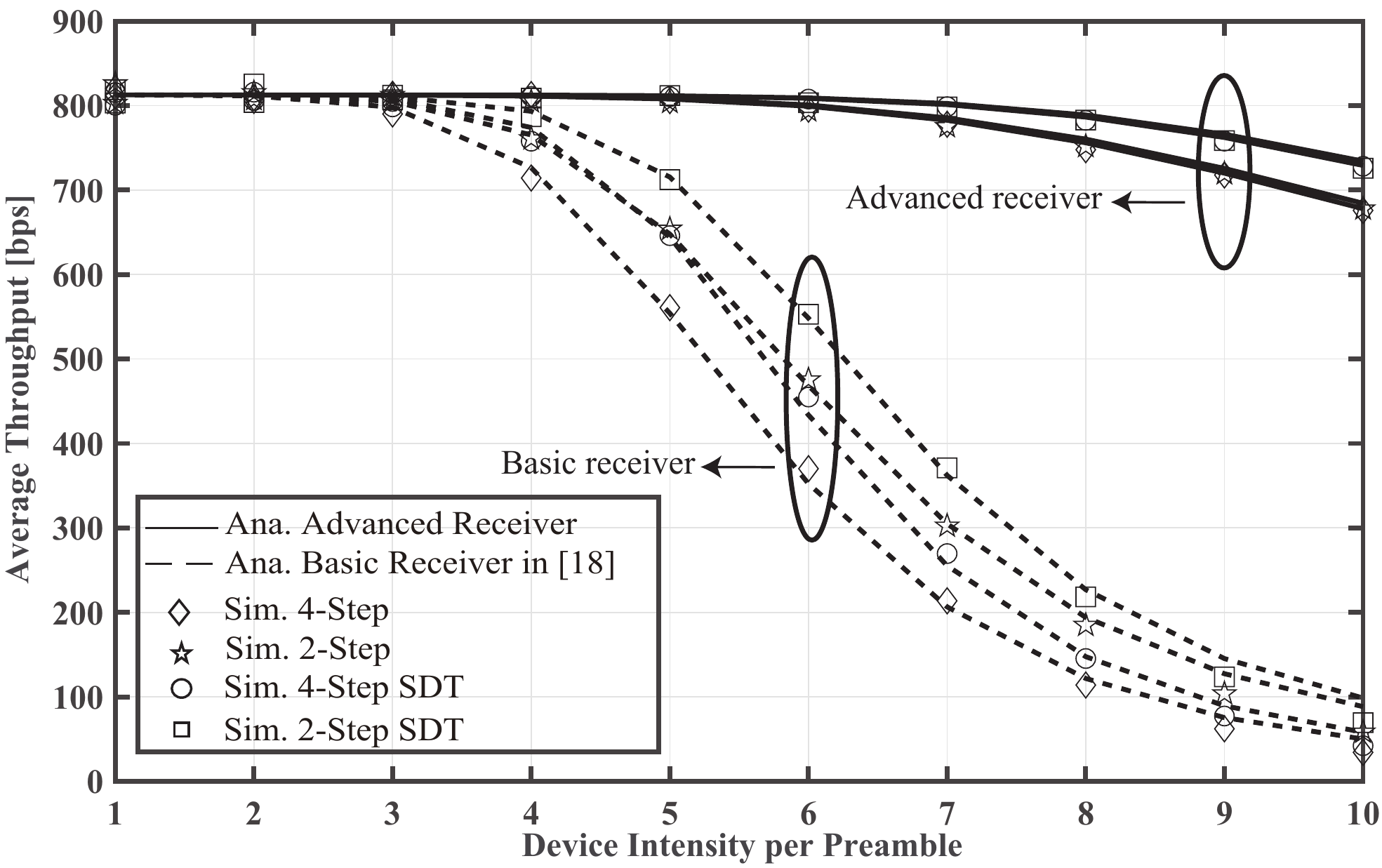}
				\caption{Throughput in $ 10 $th time slot with different device intensity}
				\label{fig:Throughput}
			\end{figure}

	Fig.~\ref{fig:Average throughput in each time slot} and Fig.~\ref{fig:Throughput} plot the throughput using \eqref{eq:throughput} with four RA schemes versus the time slot and the device intensity, respectively. In Fig.~\ref{fig:Average throughput in each time slot}, we can observe that the throughput reaches the highest value in the advanced receiver, which is determined by the new packets arrival rate during each RACH period. It also shows that the throughput of 4-step RA scheme in basic receiver rises and then decreases due to serious traffic congestion. In Fig.~\ref{fig:Throughput}, we observe that all schemes can deal with low intensity scenario effectively, and the average throughput in the basic receiver decreases more sharply. In the advanced receiver, the average throughput follows 2-step SDT $ = $ 4-step SDT$ > $ 2-step $ = $ 4-step RA schemes. In the basic receiver, the average throughput follows 2-step SDT $ > $ 2-step $ > $ 4-step SDT $ > $ 4-step RA schemes.
	\section{Conclusion}
	
	In this paper, we developed a spatio-temporal mathematical model to analyze the potential RA schemes for small data transmission. We first analyzed the preamble detection probability, PUSCH decoding probability, data transmission success probability, and overall packet transmission success probability of a typical device with 4/2-step and 4/2-step SDT RA schemes by modelling the queue evolution over consecutive time slots. We also derived average throughput for each scheme. We then provided the average energy consumption for each scheme based on the energy consumption of each message and transmission probability. Our numerical results have shown that the 2-step SDT RA scheme achieves the highest packet transmission success probability, and the lowest average energy consumption. However, the performance gain compared to the other RA schemes decreases with the increase of device intensity.

	\appendices
	\section{A Proof of Lemma1}

	We can calculate the Cumulative Distribution Function (CDF) of PDP peak value of a typical device based on \eqref{eq:power delay profile} as
		\begin{equation}
		\begin{aligned}
		&{\mathbb{P}[\mathrm{PDP}_{o}^{m}\left[ \tau_{o}\right] < {\lambda _{{\rm{th}}}}|\mathcal{A}]}\\&=\mathbb{P}\left[{ \left| {\sum\limits_{k = 0}^{{{N}_{\rm{ZC-1}}}} \left( {\sqrt {\rho }{h_0} } z_{r}^{i}\left[ k+\tau_{0}\right] + n_0\right) {z_{r}^{i}\left[ k+\tau_{0}\right]}^ * } \right|^2 < {\lambda _{{\rm{th}}}} | \mathcal{A}}\right] \\&= \mathbb{P}\left[ {\left| {{N_{\rm{ZC}}}\sqrt {\rho }{h_0}  + \sum\limits_{k = 0}^{{N_{\rm{ZC}}-1}} {\mathop {n_0}\limits^ \sim  } } \right|^2} < {\lambda _{{\rm{th}}}}|\mathcal{A}\right].
		\end{aligned}
		\end{equation}
		where the channel $h_{0} \sim \mathcal{CN}(0, 1)$, $\mathop {n_0}\limits^ \sim \sim \mathcal{CN}(0, \sigma_{n}^2)$, and $\left| {{N_{\rm{ZC}}}\sqrt {\rho }{h_0}  + \sum\limits_{k = 0}^{{N_{\rm{ZC}}-1}} {\mathop {n_0}\limits^ \sim  } } \right|^2$ follows the exponential distribution with scale $\rho N_{\rm{ZC}}^2 + \sigma_{n}^2 N_{\rm{ZC}}$. Therefore, the probability that PDP peak value of a typical device is below the detection threshold is derived as
		\begin{equation}
		\begin{aligned}
		{\mathbb{P}[\mathrm{PDP}_{j}^{m}\left[ \tau_{j}\right] < {\lambda _{{\rm{th}}}}|\mathcal{A}]}&=1-\mathrm{ exp}\left( -\dfrac{\lambda _{{\rm{th}}}}{\rho N_{\rm{ZC}}^2 + \sigma_{n}^2 N_{\rm{ZC}}}\right).
		\label{eq:PDP detecion probability}
		\end{aligned}
		\end{equation}

	Substituting \eqref{eq:PDP detecion probability} into \eqref{eq:preamble}, we can obtain the preamble detection success probability as \eqref{eq:preamble transmission success probability}.
	\section{A Proof of Lemma2}

	As the channel gain $|h|^2$ follows the exponential distribution with unit mean, the Complementary Cumulative Distribution Function (CCDF) of the maximum channel gain between $ n+1$ independent Rayleigh fading channel gains is derived as
	\begin{equation}
	F_{|h|^2_{max}}|_{N=n}\left( |h|^2 \right) =1-\left( 1-\mathrm{ exp}\left( -|h|^2 \right)  \right)^{n+1}.
	\label{eq:max_channel gain}
	\end{equation}
	
Therefore, the probability that highest SINR is above the threshold is derived as
		\begin{equation}
		\begin{aligned}
		&{\mathbb{P}[{\mathrm{SINR}}_{o} > {\gamma _{{\mathrm{th}}}}||h_{o}|^2>|h_{j}|^2 ,\mathcal{B}]}\\&= {\mathbb{E}_{{{\mathcal {I}}_{\mathrm{ intra}}}}\left\lbrace 1-\left( 1-\mathrm{ exp}\left\lbrace \dfrac{\gamma_{\rm{th}}}{\rho}\left( \sigma_{\mathrm{n}}^{2}+{{\mathcal {I}}_{\mathrm{ intra}}}\right) \right\rbrace \right)^{n+1}  \right\rbrace }.
		\end{aligned}
		\label{eq:probability_max_channel_gain}
		\end{equation}
	
	Substituting \eqref{eq:probability_max_channel_gain} into \eqref{eq:Msg3 capture effect model} and noting that $ \mathbb{P}[|h_{o}|^2>|h_{j}|^2|\mathcal{B}] = 1/(n+1)$, we obtain
	\begin{equation}
	\begin{aligned}
	&{\cal P}_{{\mathrm{pus}|\mathcal{B}}}^{m} = \\&\dfrac{{\mathbb{E}_{{{\mathcal {I}}_{\mathrm{ intra}}}}\left\lbrace 1-\left( 1-\mathrm{ exp}\left\lbrace \dfrac{\gamma_{\rm{th}}}{\rho}\left( \sigma_{\mathrm{n}}^{2}+{{\mathcal {I}}_{\mathrm{ intra}}}\right) \right\rbrace \right)^{n+1}  \right\rbrace }}{{\left( n+1\right) }}.
	\end{aligned}
	\label{eq:capture initial result}
	\end{equation}
	
	Because of the independency of the PPP in different regions and after applying the binomial expansion for the numerator of \eqref{eq:capture initial result}, we obtain
	\begin{equation}
	\begin{aligned}
	&{\cal P}_{{\mathrm{pus}|\mathcal{B}}}^{m} =\\&\dfrac{{\sum_{k=1}^{n+1}\tbinom{n+1}{k}(-1)^{k+1}\mathrm{ exp}\left\lbrace \dfrac{-k\gamma_{\rm{th}}\sigma_{\mathrm{n}}^{2}}{\rho}\right\rbrace {{\mathcal {L}}_{{{\mathcal {I}}_{\mathrm{ intra}}}}}\left({\frac {{\gamma _{\mathrm{ th}}} }{\rho } \Big |\mathcal{B} }\right)}}{{\left( n+1\right) }} ,
	\end{aligned}
	\label{eq:capture second initial result}
	\end{equation}
	where $ {{\mathcal {L}}_{{{\mathcal {I}}_{\mathrm{ intra}}}}}(\cdot) $ denotes the Laplace Transform of the aggregate intra-cell interference $ {{\mathcal {I}}_{\mathrm{ intra}}} $. The Laplace Transform of the aggregate intra-cell interference $ {{\mathcal {I}}_{\mathrm{ intra}}} $ is given in \cite{Jiang2018} as
	\begin{equation}
	{{\mathcal {L}}_{{{\mathcal {I}}_{\mathrm{ intra}}}}}\left({\frac {{\gamma _{\mathrm{ th}}} }{\rho } \Big |\mathcal{B} }\right) ={1}/{(1+\gamma _{th})^{n}}.
	\label{eq:Lap_intra interference}
	\end{equation}
	
	Substituting the Laplace Transform of the aggregated intra-cell interference \eqref{eq:Lap_intra interference} into \eqref{eq:capture second initial result}, we can obtain \eqref{eq:capture final result}

	\ifCLASSOPTIONcaptionsoff
	\newpage
	\fi

	\bibliographystyle{IEEEtran}
	
	\bibliography{IEEEabrv,manuscript}

\end{document}